\documentclass[12pt,a4paper]{article}

\usepackage[utf8]{inputenc}
\usepackage[english]{babel}
\usepackage[left=2cm,right=2cm,
    top=2.5cm,bottom=2.5cm,bindingoffset=0cm]{geometry}
    \usepackage{graphicx}
      \usepackage{amsmath,amssymb,euscript,amsthm,amsfonts,mathrsfs,amscd}
\usepackage{color}
\usepackage{amsthm}
\newtheorem{theorem}{Theorem}
\newtheorem{remark}{Remark}
\newtheorem{proposition}{Proposition}
\numberwithin{equation}{section}
\usepackage[affil-it]{authblk}
\def \rT {\rm T}
\def \bZ {\textbf{Z}}
\def \bz {\textbf{z}}

 \title{Asymptotic behaviour of weighted differential entropies in a Bayesian problem }
\author{Mark Kelbert \thanks{Electronic address: \texttt{mark.kelbert@gmail.com}}}

\author{Pavel Mozgunov%
  \thanks{Electronic address: \texttt{pmozgunov@gmail.com}; corresponding author}}
\affil{International Laboratory of Stochastic Analysis and Its Applications\\ National Research University Higher School of Economics\\ Moscow, Russia}
\affil{Department of Mathematics \\ Swansea University\\ Swansea, UK}
\date{} 
\begin{document}
\maketitle
\begin{abstract}
We consider a Bayesian problem of estimating of probability of success in a series of conditionally independent trials with binary outcomes. We study the asymptotic behaviour of differential entropy for posterior probability density function conditional on $x$ successes after $n$ conditionally independent trials, when $n \to \infty$. It is shown that after an appropriate normalization in cases $x \sim n$ $x$ $\sim n^\beta$ ($0<\beta<1$) limiting distribution is Gaussian and the differential entropy of standardized RV converges to differential entropy of standard Gaussian random variable. When $x$ or $n-x$ is a constant the limiting distribution in not Gaussian, but still the asymptotic of differential entropy can be found explicitly.

Then suppose that one is interested to know whether the coin is fair or not and for large $n$ is interested in the true frequency. To do so the concept of weighted differential entropy introduced in \cite{Belis1968} is used when the frequency $\gamma$ is necessary to emphasize. It was found that the weight in suggested form does not change the asymptotic form of  Shannon, Renyi, Tsallis and Fisher entropies, but change the constants. The main term in weighted Fisher Information is changed by some constant which depend on distance between the true frequency and the value we want to emphasize.

In third part we derived the weighted versions of Rao-Cram\'er, Bhattacharyya and Kullback inequalities.
This result is applied to the Bayesian problem described above. The asymptotic forms of these inequalities are obtained for a particular class of weight functions.
\end{abstract}



\tableofcontents

\section{Introduction}
Let ${\rm U}$ be a random variable (RV) that uniformly distributed in interval $[0,1]$. Given a realization of this RV $p$, consider a sequence of conditionally independent identically distributed $\xi_i$, where
$\xi_i=1$ with probability $p$ and $\xi_i = 0$ with probability $1-p$. Let $x_i$, each $0$ or $1$, be an outcome in trial $i$. Denote by $S_n = \xi_1+ \ldots + \xi_n$, by $\textbf{x}=(x_i,$ $i =1,...,n) $ and by $x=x(n) = \sum_{i=1}^{n} x_i$. Note that RVs $(\xi_i)$ are positively correlated. Indeed,
$P(\xi_i=1,\xi_j=1)=\int_0^1p^2 dp=1/3$ if $i\not=j$,
but $P(\xi_i=1)P(\xi_j=1)=(\int_0^1p dp)^2=1/4$.

 The probability that after $n$ trials the exact sequence $\textbf{x}$ will appear:
\begin{equation}
 \mathbb{P}(\xi_1 = x_1,...,\xi_n=x_n) = \int_0^1 p^x(1-p)^{n-x}{\rm d} p=\frac{1}{(n+1) {n\choose x}}.
 \end{equation}
This implies that the posterior probability density function (PDF) of the number of $x$ successes after $n$ trials is uniform:

$$ \mathbb{P}(S_n=x) = \frac{1}{(n+1)}, x=0, \ldots, n.$$

The posterior PDF given the information that after $n$ trials one observes $x$ successes takes the form

\begin{equation}
\displaystyle f_{p|S_n}(p|\xi_1 = x_1,...,\xi_n=x_n) = (n+1) {n\choose x} p^x (1-p)^{n-x}.
\label{pdf}
\end{equation}
Note that conditional distribution given in (\ref{pdf}) is a Beta-distribution $B(x+1,n-x+1)$. \textquotedblleft It is known that Beta-distribution is asymptotically normal with its mean and variance as $x$ and $(n-x)$ tend to infinity, but this fact is lacking a handy reference\textquotedblright (see \cite[p.1]{mit}). That is why, we give the proof of this fact in two cases.

The RV $Z^{(n)}$ with PDF (\ref{pdf}) has the following expectation:
\begin{equation}
\mathbb{E}[Z^{(n)}|S_n=x]= \displaystyle \frac{x+1}{n+2}, \end{equation}

\noindent and the following variance:
\begin{equation}\displaystyle \mathbb{V}[Z^{(n)}|S_n=x]=\frac{(x+1)(n-x+1)}{(n+3)(n+2)^2}.
\end{equation}
Recall: $h_d(f)$ is the differential entropy of some RV $Z$ with PDF $f$:
\begin{equation}
h_{d}(f) = - \int_{\mathbb{R}} f(z) {\rm log} (f(z)) {\rm d}z
\end{equation}
with convention $0 {\rm log} 0 = 0$. 
Note that after a linear transformation of RV $Z$ to RV $X$ with some PDF $g(x)$ where $X=d_1Z + d_2$ differential entropy of RV $X$ transforms in the following way \cite{Cover, Kelbert2014}:
\begin{equation}
h_d(g)=h_d(f)+ {\rm log}d_1
\label{property}
\end{equation}
Let $\bar{Z}$ be a standard Gaussian RV with PDF $\varphi$ then the differential entropy of $\bar{Z}$ \cite{Kelbert2014}:
$$h_d(\varphi) = \frac{1}{2}{\rm log} \left(2 \pi e \right).$$

\textit{The goal of the first part} of the work is to study the asymptotic behaviour of differential entropy of the following  RVs:
\begin{enumerate}
\item $Z_{\alpha}^{(n)}$ with PDF $f_{\alpha}^{(n)}$ given in (\ref{pdf}) when $x = x(n) \sim \alpha n$, where $0<\alpha<1$
\item $Z_{\beta}^{(n)}$ with PDF $f_{\beta}^{(n)}$ given in (\ref{pdf}) when $x = x(n) \sim n^\beta$, where $0<\beta<1$
\item $Z_{x}^{(n)}$ with PDF $f_{x}^{(n)}$ given in (\ref{pdf}) when  $x=c_1$ and $Z_{n-x}^{(n)}$ with PDF $f_{n-x}^{(n)}$ given in (\ref{pdf}) when $n-x(n)=c_2$ where $c_1$ and $c_2$ are some  constants. 
\end{enumerate}

We will demonstrate that the limiting distributions of standardized RV when $n \to \infty$ in the cases 1 and 2 are Gaussian. However, the asymptotic normality does not imply automatically the limiting form of differential entropy. In general the problem of taking the limits under the sign of entropy is rather delicate and was extensively studied in literature, cf., i.e., \cite{Dobrushin1960,Kelbert2010}. In the third case the limiting distribution is not Gaussian, but still the asymptotic of differential entropy can be found explicitly. 

\textit{In second part} of the paper (section 3) we suppose that one is interested to know whether the coin is fair or not and for large $n$ is interested in true frequency. So the goal of a statistical experiment in twofold: on the initial stage an experimenter is mainly concerns whether the coin is fair (i.e. $p=1/2$) or not. As the size of a sample grows, he proceeds to estimating the true value of the parameter anyway. We want to quantify the differential entropy of this experiment taking into account its two sided objective. It seems that quantitative measure of information gain of this experiment is provided by the concept of weighted differential entropy  \cite{Clim2008, Belis1968,suhov2015, suhov2015_1}. In our case $\phi(x)$ is a weight function that underline the importance of $0.5$.

The goal of the second part of work is to study the weighted Shannon (\ref{shannon}), Renyi (\ref{renyi}), Tsallis (\ref{tsal}) and Fisher (\ref{fisher}) entropies \cite{Cover}:

\begin{equation}
 h^{\phi}(f) = - \int_{\mathbb{R}} \phi^{(n)}(p) f(p) {\rm log} f(p) dp,
 \label{shannon}
 \end{equation}
 
\begin{equation}H^{\phi}_{\nu}(f) = \frac{1}{1-\nu} {\rm log} \int_{\mathbb{R}}  \phi^{(n)}(z) \left( f(z) \right)^{\nu}{\rm d}z
\label{renyi}
\end{equation}

\begin{equation}
S^{\phi}_{q}(f) = \frac{1}{q-1} \left(1- \int_{\mathbb{R}}  \phi^{(n)}(z) \left( f(z) \right)^{q}{\rm d}z \right)
\label{tsal}
\end{equation}

\begin{equation}I^{\phi}(\theta) = \mathbb{E} \left(  \phi^{(n)}(Z) \left( \frac{\partial}{\partial \theta} {\rm log} f(Z;\theta) \right)^2 { \Big | \theta  }  \right) 
\label{fisher}
\end{equation}
 \noindent where $Z=Z^{(n)}$ is a RV with PDF $f$ given in (\ref{pdf}) and $\phi^{(n)}(p)$ is a weight function that underline the importance of some particular value. The following special cases are considered:
\begin{enumerate}
\item  $\phi^{(n)}(p) = 1$ 
\item  $\phi^{(n)}(p)$ depends both on $n$ and $p$
\end{enumerate}

We will denote by $\gamma$ the frequency that we want to emphasize (the $0.5$ in the example above). We assume that $\phi(x) \geq 0$ for all $x$.
Choosing the weight function we adopt the following normalization rule:
\begin{equation}
\int_{\mathbb{R}} \phi^{(n)}(p) f^{(n)}(p) dp = 1
\label{norm}
\end{equation}

It can be easily checked that if weight function $\phi^{(n)}(p)$ satisfies (\ref{norm}) then the Renyi weighted entropy (\ref{renyi}) and Tsallis weighted entropy (\ref{tsal}) tend to Shannon's weighted entropy as $\nu \to 1$ and $q \to 1$ correspondingly.

Considering the goal of including the weight function - emphasizing some particular value, we consider the following weight function:

\begin{equation}
\phi^{(n)}(p) = \Lambda^{(n)}(\gamma) p^{\gamma \sqrt{n}} (1-p)^{(1-\gamma)\sqrt{n}},
\label{weight}
\end{equation}
where $\Lambda^{(n)}(\gamma)$ is found from the normalizing condition (\ref{norm}) and is given explicitly in (\ref{lambda}). This weight function is selected as a model example with a twofold goal to emphasize a particular value $\gamma$ for moderate $n$, while preserving the true frequency $p^*$.

\textit{In the third part} of paper (Section 4,5 and 6) we recall the statistical experiment with binary outcomes where the main objective is to find out whether the probabilities of success and failure are equal. In other words, the statistical decisions in a neighbourhood of a particular value $\gamma=1/2$ are especially sensitive. It is clear that if an experimenter  wrongly declares that the parameter of interest is in a small neighbourhood of particular value $\gamma=1/2$ than the penalty of this error should be more severe than for a similar error far from the sensitive area. Similar models of sensitive estimator appear in many fields of statistics. For this reason we start with the general framework and then specialize it to the case of binary trials as an example. 

Consider RV $\bZ \in \mathbb{R}^d$ with PDF $f(\bz)$ or family of RV $\bZ_\theta \in \mathbb{R}^d$ with PDF $f_\theta$ where $\theta \in \Theta \subset \mathbb{R}^m$ is the vector of parameters of PDF $f_\theta$. Denote $\bz = \left[z_1, \ldots, z_d \right]^{\rT}$. Let $\phi(.)$ be the positive weight function that emphasizes particular value $\gamma$ ,  $\mathbb{E}^\phi_\theta(\bZ)$ be the weighted expectation of random vector $\bZ$ with PDF $f_\theta$
\begin{equation}
 g(\theta) \equiv  \mathbb{E}^\phi_\theta(\bZ) = \int_{\mathbb{R}^d} \bz f_\theta(\bz) \phi(\bz) {\rm d}\bz
 \label{expect}
\end{equation}
and $\mathbb{E}_\theta(\bZ)$ be the classic expectation of random vector $\bZ$ with PDF $f_\theta$ 
\begin{equation}
 e(\theta) \equiv  \mathbb{E}_\theta(\bZ) = \int_{\mathbb{R}^d} \bz f_\theta(\bz)  {\rm d}\bz.
 \label{expectclasic}
\end{equation}

 Quantitative measures of information gain of experiments of the type described above are provided by the weighted Shannon differential entropy \cite{Belis1968,suhov2015, suhov2015_1}
 \begin{equation}
 h^{\phi}(f_\theta)= - \int_{\mathbb{R}^d}\phi(\bz)f_\theta(\bz) \log f_\theta(\bz){\rm d}\bz,
 \label{shannon}
 \end{equation}
the weighted ($m \times m$) Fisher information matrix
 \begin{equation}{\rm I}^{\phi}(\theta) = \mathbb{E}_\theta^\phi  \left[ \left( \frac{\partial}{\partial \theta} {\rm log} f_\theta(\bZ) \right) \left( \frac{\partial}{\partial \theta} {\rm log} f_\theta(\bZ) \right)^{\rT} \right]
\label{fisher}
\end{equation}
where $\frac{\partial}{\partial \theta}$ is the notation for the gradient (the vector $\frac{\partial}{\partial \theta} {\rm log} f_\theta(\bZ)$ is the \textit{score}), and the weighted Kullback-Leibler divergence of $g$ from $f$ \cite{Kelbert2015}
\begin{equation}
\mathbb{D}^\phi(f||g) = \int_{\mathbb{R}^d} \phi(\bz)f(\bz) {\rm log} \frac{f(\bz)}{g(\bz)} {\rm d}\bz.
\label{kullback}
\end{equation}
For simplicity we assume that the inverse Fisher matrix exists, but, in a general case, under inverse we understand the Moore-Penrose pseudoinverse. Also it is shown that in this context it is more convenient to study the calibrated  Kullback-Leibler divergence defined in \cite{Kelbert2015}:
\begin{equation}
K^\phi(f||g)= \int_{\mathbb{R}^d} \phi(\bz) \frac{f(\bz)}{C(f)} {\rm log} \frac{f(\bz) C(g)}{g(\bz)C(f)} {\rm d} \bz = \mathbb{D}(\tilde{f}||\tilde{g})
\end{equation}
where $C(f) = \int_{\mathbb{R}^d} \phi(\bz) f(\bz) {\rm d}\bz$,
$\tilde{f} = {\phi(\bz)f(\bz)} C(f)^{-1}$ and $\mathbb{D}(f||g)$ is the standard Kullback-Leibler divergence of $g$ from $f$
\begin{equation}
\mathbb{D}(f||g) = \int_{\mathbb{R}^d} f(\bz) {\rm log} \frac{f(\bz)}{g(\bz)} {\rm d}\bz.
\end{equation}
The goal of the third part is twofold. Firstly, the weighted analogous of the Rao-Cram\'er, Bhattacharyya and Kullback inequalities will be derived in a general case. Secondly, these inequalities will be illustrated in the example  described above which has an independent interest.

\section{Asymptotic of Shannon's differential entropy}
\begin{theorem}
Let $\tilde{Z}_{\alpha}^{(n)}=n^{\frac{1}{2}} (\alpha(1-\alpha))^{-\frac{1}{2}} (Z_{\alpha}^{(n)} - \alpha) $ 
 be a RV with PDF $\tilde{f}_{\alpha}^{(n)}$. Let $\bar{Z} \sim  \mathcal{N}(0,1)$ be the standard Gaussian RV, then

\textbf{(a)} $\tilde{Z}_{\alpha}^{(n)}$ weakly converges to $\bar{Z}$:
$$ \tilde{Z}_{\alpha}^{(n)} \Rightarrow \bar{Z} \ {\rm as} \ n \to \infty. $$
\textbf{(b)} The differential entropy of $\tilde{Z}_{\alpha}^{(n)}$converges to differential entropy of $\bar{Z}$:
$$\lim_{n \rightarrow \infty} h(\tilde{f}_{\alpha}^{(n)})= \frac{1}{2} {\rm log} \left(2 \pi e \right).$$
\textbf{(c)} The Kullback-Leibler divergence of $\varphi$ from $\tilde{f}_{\alpha}^{(n)}$ tends to $0$ as $n \to \infty$:
$$\lim_{n \rightarrow \infty} \mathbb{D}(\tilde{f}_{\alpha}^{(n)}||\varphi)=0.$$
\label{theorem1}
\end{theorem}
\begin{proof}
\textbf{(a)} Let $x=x(n)=\alpha n$ where $0<\alpha<1$ and consider RV $$\tilde{Z}_\alpha^{(n)}=n^{\frac{1}{2}} (\alpha(1-\alpha))^{-\frac{1}{2}} (Z_{\alpha}^{(n)} - \alpha) .$$ We proceed by the method of characteristic functions, and establish that:

\begin{equation}
\phi(t)=\mathbb{E}[e^{it\tilde{Z}_{\alpha}^{(n)}}]\to e^{-t^2/2}
\label{gaussian}
\end{equation}
for all $t \in \mathbb{R}$.
Indeed

$ \displaystyle \phi(t) = \int_0^1e^{it\frac{(p-\alpha)\sqrt{n}}{\sqrt{\alpha(1-\alpha)}}}f^{(n)}_{\alpha}(p){\rm d} p = (n+1) {n\choose x}  e^{it\frac{(-\alpha)\sqrt{n}}{\sqrt{\alpha(1-\alpha)}}}\int_0^1e^{it\frac{p\sqrt{n}}{\sqrt{\alpha(1-\alpha)}}}p^x (1-p)^{n-x}{\rm d} p $

\noindent and consider the integral:
\begin{equation}
I(t,\alpha,n)= \int_0^1e^{n(it\frac{p}{\sqrt{\alpha(1-\alpha)n}}+\alpha  {\rm log}p + (1-\alpha) {\rm log}(1-p))}{\rm d} p.
\label{integral11}
\end{equation}
Denote $g(p)=it\frac{p}{\sqrt{\alpha(1-\alpha)n}}+\alpha  {\rm log}p + (1-\alpha) {\rm log}(1-p)$. The integrand in (\ref{integral11}) has a narrow sharp peak, and the integral is
completely dominated by the maximum of ${\rm Re}[g(p)]$ when $n \to \infty$.
For fixed values of $t$, $\alpha$ and $n \to \infty$, it can be studied by the saddle point method \cite[Theorem 1.3, p.170]{Fedoruk1977}:
\begin{equation}
 I(t,\alpha,n) \simeq e^{ng(p^*)}\sqrt{\frac{ 2 \pi}{-n g''(p^*)}} \left(1+ O \left(\frac{1}{n} \right)\right).\end{equation}
Find the point of maximum of ${\rm Re}[g(p)]$ and deform initial contour $[0,1]$ into the steepest descent contour through the saddle point:
$$p^{*} = \alpha + it \frac{\sqrt{(1-\alpha)\alpha}}{\sqrt{n}} + O \left(\frac{1}{n} \right).$$
So, $\phi(t)$ takes the form:
$$ \phi(t) = e^{- t^2} \displaystyle (n+1) {n\choose x}(p^*)^x (1-p^*)^{n-x} \sqrt{\frac{ 2 \pi}{-n g''(p^*)}} + O \left(\frac{1}{n} \right). $$
Here and below $x= \lfloor \alpha n \rfloor$. Next, by Stirling's formula:
$$ (n+1) {n\choose x} \simeq (n+1) \frac{n^n}{x^x (n-x)^{(n-x)}} \sqrt{\frac{n}{2 \pi x (n-x)}}.$$
So, the straightforward computation yields:
$$ \begin{array}
 {l} \displaystyle (p^*)^x (1-p^*)^{n-x}\simeq \alpha^x (1-\alpha)^{(n-x)} e^{it \sqrt{(1-\alpha)\alpha n} + \frac{(1-\alpha)t^2}{2} -it \sqrt{(1-\alpha)\alpha n} + \frac{\alpha t^2}{2}}= \\ \displaystyle = e^{\frac{t^2}{2}} \left(\frac{x}{n} \right) ^x \left(\frac{n-x}{n}\right)^{n-x}. \end{array}$$
It can be checked that next term in asymptotic of ${\rm log}p^*$ (as well as ${\rm log}(1-p^*)$) is decaying to $0$ after multiplication of $\alpha n$ and $(1-\alpha)n$, correspondingly.

We have for $t \in \mathbb{R}$
$$ \begin{array}
 {l} \displaystyle \phi (t) \simeq  e^{- t^2} \frac{(n+1)n^n}{x^x (n-x)^{(n-x)}} \sqrt{\frac{n}{2 \pi x (n-x)}}e^{\frac{t^2}{2}} \left(\frac{x}{n} \right) ^x \left(\frac{n-x}{n}\right)^{n-x}\sqrt{\frac{ 2 \pi x (n-x)}{n^3}}\simeq \\ \displaystyle \simeq  e^{-\frac{t^2}{2}} \end{array}.$$
This fact establishes pointwise convergence of characteristic function to its Gaussian limit and it completes the proof of part (a).

\textbf{(b)}
Write the differential entropy in the form:
\begin{equation}
\displaystyle h(f_{\alpha}^{(n)})= -\left( {\rm log} \left[(n+1) {n\choose x}\right] + (n+1) {n\choose x} x I_1 + (n+1) {n\choose x} (n-x) I_2 \right)\end{equation}
where
\label{hdiff}
\begin{equation}
I_1 = \int_0^1  p^x (1-p)^{n-x} {\rm log} p {\rm d} p,
\label{int1}
\end{equation}
\begin{equation}
I_2 =\int_0^1  p^x (1-p)^{n-x} {\rm log} (1-p) {\rm d} p.
\label{int2}
\end{equation}
Integrals $I_1$ and $I_2$ can be computed explicitly by reducing to the standard integral
\begin{equation}
\int_0^1x^{\mu-1}(1-x^r)^{\nu-1}{\rm log}x{\rm d} x = \frac{1}{r^2} B\left(\frac{\mu}{r},\nu \right) \left(\psi\left(\frac{\mu}{r}\right) - \psi\left(\frac{\mu}{r}+\nu\right)\right)
\label{ryz}
\end{equation}
\noindent where $\psi(x)$ is the digamma function, and $B(x,y)$ is the Beta-function \cite[\#4.253.1]{rg2007} and in considering case $r\equiv 1, \mu-1 \equiv x, \nu-1 \equiv n-x$.

For integral $I_1$, we get:
$$ U_1 = (n+1) {n\choose x} x I_1 = -x(\psi(n+2)-\psi(x+1)).$$
Similarly, for the second integral $I_2$, we obtain:
$$ U_2 = (n+1) {n\choose x} (n-x) I_2 = -(n-x)(\psi(n+2)-\psi(n-x+1)).$$
After summation of these two integrals and using the asymptotic for digamma function \cite[\#8.362.2]{rg2007}, we obtain:
$$U_1+U_2 = x{\rm log}x - n{\rm log}n + (n-x){\rm log}(n-x) - \frac{1}{2}+O \left(\frac{1}{n} \right).  $$
Next, we apply the Stirling formula to the first term in (\ref{hdiff}):
$$ \begin{array} {l} \displaystyle U_0 = {\rm log} \left[(n+1) {n\choose x}\right] = n{\rm log}n -x{\rm log}x -(n-x){\rm log}(n-x) + \\ + \displaystyle  \frac{1}{2}{\rm log}n  - \frac{1}{2}{\rm log}\alpha - \frac{1}{2}{\rm log}(1-\alpha) -{\rm log}(\sqrt{2 \pi}) + O \left(\frac{1}{n} \right). \end{array}.$$
Here as before $x= \lfloor \alpha n \rfloor$. So, we obtain the following asymptotic of the differential entropy:
\begin{equation}
 \lim_{n \rightarrow \infty} \left[h(f_{\alpha}^{(n)})-\frac{1}{2} {\rm log}\frac{2 \pi e [\alpha (1-\alpha)]}{n} \right] = 0.
\label{res1}
\end{equation}
Due to (\ref{property}), the differential entropy of RV $\tilde{Z}_\alpha^{(n)}$ has the form:
\begin{equation}
 \lim_{n \rightarrow \infty} \left[h(\tilde{f}_{\alpha}^{(n)}) \right] = \frac{1}{2} {\rm log} \left( 2 \pi e \right).
\end{equation}
\textbf{(c)} By the definition of the the Kullback-Leibler divergence:
$$ \mathbb{D}(\tilde{f}_{\alpha}^{(n)}||\varphi)= - h(\tilde{f}_{\alpha}^{(n)}) - \int_0^1 \tilde{f}_{\alpha}^{(n)}(p) \log \varphi(p) {\rm d}p $$
$$= - \frac{1}{2} {\rm log} \left( 2 \pi e \right) + \frac{1}{2} \log (2 \pi) + \frac{1}{2} \int_0^1 p^2 \tilde{f}_{\alpha}^{(n)} {\rm d}p + O \left( \frac{1}{n} \right) =  O \left( \frac{1}{n} \right),   $$
 $\int_0^1 p^2 \tilde{f}_{\alpha}^{(n)} {\rm d}p = 1 +O \left( \frac{1}{n} \right) $ is the second moment of $\tilde{Z}_\alpha^{(n)}$. It completes the proof.
 \end{proof}

\begin{theorem}

Let $\tilde{Z}_{\beta}^{(n)}=n^{1-\beta/2}(Z_{\beta}^{(n)}-n^{\beta-1})$ be a RV with PDF $\tilde{f}_{\beta}^{(n)}$ and $\bar{Z} \sim  \mathcal{N}(0,1)$ then

\textbf{(a)} $\tilde{Z}_{\beta}^{(n)}$ weakly converges to $\bar{Z}$:
$$\tilde{Z}_{\beta}^{(n)} \Rightarrow \bar{Z} \ {\rm as} \ n \to \infty. $$
\textbf{(b)} The differential entropy of $\tilde{Z}_{\beta}^{(n)}$ converges to differential entropy of $\bar{Z}$:
 $$ \lim_{n \rightarrow \infty} h(\tilde{f}_{\beta}^{(n)})=\frac{1}{2} {\rm log} \left(2 \pi e \right). $$
 \textbf{(c)} The Kullback-Leibler divergence of $\varphi$ from $\tilde{f}_{\beta}^{(n)}$ tends to $0$ as $n \to \infty$:
$$\lim_{n \rightarrow \infty} \mathbb{D}(\tilde{f}_{\beta}^{(n)}||\varphi)=0.$$
\label{theorem2}
\end{theorem}

\begin{proof}
\textbf{(a)} Let $x=x(n)=n^\beta$ where $0<\beta<1$ and consider $\tilde{Z}_\beta^{(n)}$ such that
$$ \tilde{Z}_{\beta}^{(n)}=n^{1-\beta/2}(Z_{\beta}^{(n)}-n^{\beta-1}).$$
In this case, it is more convenient to proceed by the method of moments. We use the following classical result. Let $f_n$ be a sequence of distribution functions with finite moments $\mu_k(n)$, and $\mu_k(n)$ tends to $\nu_k$ for each $k$ as $n \to \infty$ where $\nu_k$ are moments of distribution $f$, and the distribution $f$ is uniquely defined by its moments, then $f_n$ weakly converges to $f$ as $n \to \infty$ \cite{moments}.

Consider RV $ \displaystyle \tilde{Z}_\beta^{(n)} =n^{1-\beta /2}(Z_\beta^{(n)}-n^{\beta -1})$ where $Z_\beta^{(n)}$ has PDF (\ref{pdf}) when $x= \lfloor n^\beta \rfloor$ and compute all moments of $\tilde{Z}_\beta^{(n)}$. 
First, $\mathbb{E}(\tilde{Z}_\beta^{(n)}) \to 0$ as $n \to \infty$ because
$\mathbb{E}(Z_\beta^{(n)})= n^{\beta-1}+ O \left(\frac{1}{n} \right)$. 
Next, we check that $ \mathbb{E}\left[\left(\tilde{Z}_\beta^{(n)}\right)^2\right]=n^{2(1-\beta/2)} \mathbb{E}(Z_\beta^{(n)}-n^{\beta -1})^2 \to 1$ as $n \to \infty$.
Compute central moments for any $k>1$:
\begin{equation}
\mathbb{E}\left[\left(\tilde{Z}_\beta^{(n)}\right)^k\right] = n^{k-\frac{\beta k}{2}} (1-n^{1-\beta})^{-k} (1-n^{\beta-1})^k \  _2F_1[-k,n^{\beta}+1;n+2;n^{1-\beta}] 
\label{moments}
\end{equation}
where $_2F_1[-k,n^{\beta}+1;n+2;n^{1-\beta}]$ is the hypergeometric function, which, in this case, is the polynomial:
$$ _2F_1[-k,n^{\beta}+1;n+2;n^{1-\beta}] = 
\sum_{i=0}^{k} (-1)^i {k\choose i} \frac{(n^\beta+1)_i}{(n+2)_i}n^{i(1-\beta)}$$ where
$(q)_n$ is the rising Pochhammer symbol. For $n>0$
$$(q)_n= q(q+1)...(q+n-1)$$ \noindent and $(q)_0=1$.

Consider asymptotic of terms separately:
$$n^{k-\frac{\beta k}{2}} (1-n^{1-\beta})^{-k} (1-n^{\beta-1})^k \simeq O(n^{\frac{k \beta}{2}})$$
\noindent and
\begin{equation} _2F_1[-k,n^{\beta}+1;n+2;n^{1-\beta}] \simeq O(n^{-[0.5+0.5k] \beta}) 
\label{asymp}
\end{equation}
 where
$\lfloor k \rfloor$ is the integer part of $k$. 
For $k$ odd:

\begin{equation} n^{k(1-\beta/2)} \mathbb{E}(Z_\beta^{(n)}-n^{\beta -1})^k = O(n^{\frac{k \beta}{2}}) O(n^{-[0.5+0.5k] \beta}) \simeq O(n^{-\beta /2 }) \to 0
\label{odd}
\end{equation} as $n \to \infty$. For $k$ even:

\begin{equation} n^{k(1-\beta/2)} \mathbb{E}(Z_\beta^{(n)}-n^{\beta -1})^k = O(n^{\frac{k \beta}{2}}) O(n^{-[0.5+0.5k] \beta}) = O(1).
\label{even}
\end{equation}

We see that every even central moment tends to a constant which is the coefficient in front of term $n^{-[0.5+0.5k] \beta}$ in the hypergeometric function. For $k$ even, we have:
\begin{equation} \displaystyle n^{k(1-\beta/2)} \mathbb{E}(Z_\beta^{(n)}-n^{\beta -1})^k \to (k-1)!.
\label{even}
\end{equation}
These imply that RV $\tilde{Z}_\beta^{(n)}$ weakly converges to the standard Gaussian RV.

\textbf{(b)}
Write the differential entropy in the form:

\begin{equation} \begin{array} {l} \displaystyle h(f_{\beta}^{(n)})= -\left( {\rm log}\left[(n+1) {n\choose x} \right] + (n+1) {n\choose x} x I_1 + (n+1) {n\choose x} (n-x) I_2 \right)= \\ \displaystyle = -(U_0+U_1 + U_2) \end{array}
\label{hdiff22}
\end{equation}
where $I_1$ and $I_2$ are defined in (\ref{int1}) and (\ref{int2}) and can be computed explicitly by (\ref{ryz}).

As before, we apply the Stirling formula for $U_0$:
$$\displaystyle \begin{array} {l} U_0 =   n{\rm log}n -x{\rm log}x -(n-x){\rm log}(n-x) + {\rm log}n \\ \displaystyle  + \frac{1}{2} (- {\rm log} n^{\beta} - {\rm log}(1-n^{\beta-1})) - \frac{1}{2} {\rm log}( 2 \pi)+ O \left(\frac{1}{n^\beta}  \right) \end{array}.$$
As far as $0<\beta<1$ the reminder tends to $0$ as $n \to \infty$. Note that the rate of decaying depends on parameter $\beta$, contrary to reminder in Theorem 1. Now
$U_1+U_2$ can be computed as follows:
$$\displaystyle U_1+U_2 = x{\rm log}x - n{\rm log}n + (n-x){\rm log}(n-x) - \frac{1}{2} + O \left(\frac{1}{n^\beta}  \right). $$
So, we proved that 
 $$\lim_{n \rightarrow \infty} \left[h(f_{\beta}^{(n)})-\frac{1}{2} {\rm log}\frac{2 \pi e (1-n^{\beta-1})}{n^{2-\beta}} \right] = 0 $$
Due to (\ref{property}), the differential entropy of RV $\tilde{Z}_{\beta}^{(n)}$ has the form:
  $$ \lim_{n \rightarrow \infty} h(\tilde{f}_{\beta}^{(n)})=\frac{1}{2} {\rm log} \left(2 \pi e \right). $$
 \textbf{(c)} Similarly, by the definition of the the Kullback-Leibler divergence:
$$ \mathbb{D}(\tilde{f}_{\beta}^{(n)}||\varphi)= - h(\tilde{f}_{\beta}^{(n)}) - \int_0^1 \tilde{f}_{\beta}^{(n)}(p) \log \varphi(p) {\rm d}p $$
$$= - \frac{1}{2} {\rm log} \left( 2 \pi e \right) + \frac{1}{2} \log (2 \pi) + \frac{1}{2} \int_0^1 p^2 \tilde{f}_{\beta}^{(n)} {\rm d}p + O \left( \frac{1}{n^{\beta}} \right) =  O \left( \frac{1}{n^{\beta}} \right),   $$
$\int_0^1 p^2 \tilde{f}_{\beta}^{(n)} {\rm d}p = 1 +O \left( \frac{1}{n^{\beta}} \right) $ is the second moment of $\tilde{Z}_\beta^{(n)}$.

\end{proof}

\begin{theorem}
Let $\tilde{Z}_{c_1}^{(n)}=n{Z}_{c_1}^{(n)}$ be a RV with PDF  $\tilde{f}_{{c_1}}^{(n)}$ and $\tilde{Z}_{n-c_2}^{(n)}=n{Z}_{n-c_2}^{(n)}$ be a RV with PDF $\tilde{f}_{n-c_2}^{(n)}$. Denote  $H_k =1 +\frac{1}{2} + \ldots + \frac{1}{k}$ the partial sum of harmonic series and  $\gamma$ the Euler-Mascheroni constant, then
$$ \textbf{(a)} \  \lim_{n \rightarrow \infty} h(\tilde{f}_{c_1}^{(n)})  = c_1 + \sum_{i=0}^{c_1-1} {\rm log}(c_1-i)- c_1(H_{c_1} - \gamma)+1.$$
$$ \textbf{(b)} \  \lim_{n \rightarrow \infty} h(\tilde{f}_{n-c_2}^{(n)})  = c_2 + \sum_{i=0}^{c_2-1} {\rm log}(c_2-i)- c_2(H_{c_2} - \gamma)+1.$$

\label{theorem3}

\end{theorem}

\begin{proof}

\textbf{(a)}
Let $x=c_1$ where $c_1$ is a some integer constant.
Consider the differential entropy:
$ h(f_{c_1}^{(n)})= -(U_0+U_1+U_2)$ where $U_0$, $U_1$ and $U_2$ defined in (\ref{hdiff22}).
Applying the Stirling formula for $U_0$:
$$U_0 = {\rm log}n - {\rm log}(x!) + x{\rm log}n + O \left(\frac{1}{n} \right).$$
Next, we compute $U_1+U_2$ via formula (\ref{ryz}) as before. The only difference will be in asymptotic of digamma functions \cite[\#8.365.3, \#8.365.4]{rg2007}, because of $x=c_1$ where $c_1$ is constant:

$\psi(n-x+1)\simeq {\rm log}n + \frac{1/2-x}{2n}$, and $\psi(x+1)=H_x - \gamma$, here $H_x$ is the partial sum of harmonic series and $\gamma$ stands for the Euler-Mascheroni constant.
Using that $x=c_1$:
$$ \lim_{n \rightarrow \infty} \left[h(f_{c_1}^{(n)})+{\rm log}n \right] =c_1 + \sum_{i=0}^{c_1-1} {\rm log}(c_1-i)- c_1(H_{c_1} - \gamma)+1.$$
Due to (\ref{property}) it can be written in the following form:
$$   \lim_{n \rightarrow \infty} h(\tilde{f}_{c_1}^{(n)})  = c_1 + \sum_{i=0}^{c_1-1} {\rm log}(c_1-i)- c_1(H_{c_1} - \gamma)+1.$$
\textbf{(b)}
Let $n-x(n)=c_2$ where $c_2$ is some integer constant. In a similar way we compute $h(f_{n-c_2}^{(n)})$ where $n-x=c_2$ and $c_2$ is a constant. The asymptotic of digamma function is given as follows \cite[\#8.365.4]{rg2007}:
$$\psi(n-x+1)= H_{c_2} - \gamma \ {\rm where} \ x=n-c_2, $$
and the final result for differential entropy:
$$  h(f_{n-c_2}^{(n)}) =  -{\rm log}n + c_2-c_2(H_{c_2}-\gamma) + \sum_{i=0}^{c_2-1} {\rm log}(c_2-i)+1 + O \left(\frac{1}{n} \right) .$$
In terms of standardized RV $\tilde{Z}^{(n)}_{n-c_2}$ we obtain  due to (\ref{property}):
$$  \lim_{n \rightarrow \infty} h(\tilde{f}_{n-c_2}^{(n)})  = c_2 + \sum_{i=0}^{c_2-1} {\rm log}(c_2-i)- c_2(H_{c_2} - \gamma)+1.$$
\end{proof}

\section{Asymptotic of weighted differential entropies}
The normalizing constant in the weight function (\ref{weight}) is found from the condition (\ref{norm}). We obtain that:
\begin{equation} \Lambda^{(n)}(\gamma)=\frac{\Gamma{(x+1)}\Gamma{(n-x+1)}\Gamma{(n+2+\sqrt{n})}}{\Gamma{(x+\gamma \sqrt{n}+1)}\Gamma{(n-x+1+ \sqrt{n} - \gamma \sqrt{n})} \Gamma{(n+2)}}.
\label{lambda}
\end{equation}
We denote by $\psi^{(0)}(x) = \psi(x)$ and by $\psi^{(1)}(x)$ the digamma function and its first derivative respectively.
\begin{equation}  \psi^{(n)}(x) = \frac{{\rm d^{n+1}}}{{\rm d}x^{n+1}} {\rm log} 
\left( \Gamma(x) \right)
\end{equation}
In further calculations we will need the asymptotic of these functions:
$$ \displaystyle \psi(x) = {\rm log}(x) - \frac{1}{2x} + O \left(\frac{1}{x^2} \right) \ {\rm as} \ x \to \infty,$$
$$ \displaystyle \psi^{(1)}(x) = \frac{1}{x} + \frac{1}{2x^2} +O \left(\frac{1}{x^3} \right) \ {\rm as} \ x \to \infty.$$

\begin{proposition}
Let $Z^{(n)}$ be a RV with  $f^{(n)}$ - conditional PDF after $n$ trials given by (\ref{pdf}), $h^{\phi}(f_{\alpha}^{(n)})$ - the weighted Shannon entropy of $Z^{(n)}$ given in (\ref{shannon}) . When $x=\alpha n$ {\rm (}$0<\alpha<1${\rm )}  and the weight function $\phi^{(n)}(p)$ is given in (\ref{weight})
\begin{equation}
\lim_{n \rightarrow \infty} \left( h^{\phi}(f_{\alpha}^{(n)})- \frac{1}{2} {\rm log} \left( \frac{2 \pi e \alpha (1-\alpha)}{n} \right) 
\right) = \frac{(\alpha-\gamma)^2}{2\alpha(1-\alpha)}.
\end{equation}
If the $\alpha = \gamma$ then the asymptotic of  $h^{\phi}(f)$ is exactly the asymptotic of differential Shannon's entropy with $\phi^{(n)}(p)=1$.
\end{proposition}
\begin{proof}
The Shannon differential entropy of PDF $f^{(n)}(p)=f(p)$ given in (\ref{pdf}) and weight function $\phi^{(n)}(p)$ given in (\ref{weight}) takes the form:
$$ h^{\phi}(f) = {\rm log} \left[(n+1) {n\choose x} \right] + x \int_0^1 {\rm log}(p) \phi^{(n)}(p) f(p) dp + (n-x)\int_0^1 {\rm log}(1-p) \phi^{(n)}(p) f(p) dp  $$
The integrals can be computed explicitly \cite{rg2007} (page 552):

$$
\int_0^1x^{\mu-1}(1-x^r)^{\nu-1}{\rm log}(x){\rm d} x = \frac{1}{r^2} \mathbb{B}\left(\frac{\mu}{r},\nu \right) \left(\psi\left(\frac{\mu}{r}\right) - \psi\left(\frac{\mu}{r}+\nu\right)\right),
$$

\noindent Applying this formula for integral, we get:

$ \displaystyle \int_0^1 {\rm log}(p) \phi^{(n)}(p) f(p) dp = \psi(x+z+1) - \psi(n+\sqrt{n} +2)$,
 where $z= \gamma \sqrt{n}$ and $\psi(x)$ is a digamma function.

$ \displaystyle \int_0^1 {\rm log}(1-p) \phi^{(n)}(p) f(p) dp = \psi(n-x+\sqrt{n} - z +1) - \psi(n+ \sqrt{n} +2)$

So we have that

$ \displaystyle h^{\phi}(f) ={\rm log} \left[(n+1) {n\choose x} \right] + x \psi(x+z+1) + (n-x)\psi(n-x+\sqrt{n} - z +1)  - n\psi(n+ \sqrt{n} +2)$.

By Stirling's formula we have that for $x=\alpha n$:

$\displaystyle  {\rm log} \left[(n+1) {n\choose x} \right]= n{\rm {\rm log}}(n) -x{\rm log}x -(n-x){\rm log}(n-x) + \frac{1}{2}{\rm log}(n)  - \frac{1}{2}{\rm log}(\alpha) - \frac{1}{2}{\rm log}(1-\alpha) -{\rm log}\sqrt{2 \pi} + O\left(\frac{1}{n} \right)$

Using the asymptotic for digamma function

$$ \displaystyle \psi(x+z+1)=  {\rm log}(x)+ \frac{\gamma \sqrt{n}}{x} + \frac{\alpha - \gamma^2}{2\alpha x} + O\left(\frac{1}{n^{3/2}} \right) $$
$$ \displaystyle \psi(n-x+\sqrt{n} - z +1)) = {\rm log}(n-x) + \frac{(1-\gamma)\sqrt{n}}{n-x} + \frac{2\gamma - \gamma^2-\alpha}{2(1-\alpha)(n-x)} + O\left(\frac{1}{n^{3/2}} \right) $$
$$ \displaystyle \psi(n+\sqrt{n} +2 ) = {\rm log}(n) + \frac{\sqrt{n}}{n}  + O\left(\frac{1}{n} \right),$$
 we get
\begin{equation}
 h_{\phi}^w(f^{(n)}) =  \frac{1}{2} {\rm log}\frac{2 \pi e [\alpha (1-\alpha)]}{n} + \frac{(\alpha-\gamma)^2}{2\alpha (1-\alpha)} + O\left(\frac{1}{n} \right)
 \label{hdifw1}
 \end{equation}
The first term in (\ref{hdifw1}) is differential entropy with weight $\phi \equiv 1$ of Gaussian RV. Moreover, note that the asymptotic of the weighted entropy exceeds  classical entropy studied above.  The only difference is constant, which tend to zero if $\gamma \to \alpha$.
\end{proof}

\begin{theorem}
Let $Z^{(n)}$ be a RV with  $f^{(n)}$ - conditional PDF after $n$ given by (\ref{pdf}) and with weighted Renyi differential entropy $H_{\nu}(f^{(n)})$ given in (\ref{renyi}).

\textbf{(a)} When both $(x)$ and $(n-x)$ tend to infinity as $n \to \infty$ in the case $\phi^{(n)}(p) = 1,$
\begin{equation}
\lim_{n \rightarrow \infty} \left(H_{\nu}(f^{(n)})-\frac{1}{2} {\rm log} \frac{2 \pi x(n-x)}{n^3} \right) = -  \frac{{\rm log}(\nu)}{2(1-\nu)}.
\end{equation}
For any fixed $n$ when $\nu \to 1$ Renyi's differential entropy of $Z^{(n)}$ tends to Shannon's differential entropy of $Z^{(n)}$.

\textbf{(b)} When $x=\alpha n$ {\rm (}$0<\alpha<1${\rm )} and the weighted function is given in (\ref{weight})
\begin{equation}
 \lim_{n \rightarrow \infty} \left(H^{\phi}_{\nu}(f^{(n)}_{\alpha}) - \frac{1}{2} {\rm log} \frac{2 \pi \alpha (1-\alpha)}{n} \right) =  -\frac{{\rm log}(\nu)}{2(1-\nu)} + \frac{(\alpha-\gamma)^2}{2\alpha (1-\alpha)\nu}. 
\end{equation}
For any fixed $n$ the Renyi weighted differential entropy tends to Shannon's weighted differential entropy RV with PDF given in (\ref{pdf}) as $\nu \to 1$.
\label{theorem1}
\end{theorem}
\begin{proof}

\textbf{(a)} 
In this case $\phi^{(n)}(p) \equiv 1$, so the Renyi entropy have the form:

$\displaystyle (1-\nu) H_{\nu}(f) = {\rm log} \int_{0}^1 \left( f(p) \right)^{\nu}{\rm d}p = \nu {\rm log} \left[(n+1) {n\choose x} \right]  + {\rm log} \left[ \int_0^1 p^{\nu x} (1-p)^{\nu(n-x)} \right] = U_0 + U_1  $

By Stirling formula:

$\displaystyle U_0 = \nu {\rm log} \left[(n+1) {n\choose x} \right] = 
\nu n {\rm log}(n) - \nu x {\rm log} (x) - \nu (n-x) {\rm log} (n-x) + \nu {\rm log} (n)   + \frac{\nu}{2} {\rm log}(n) - \frac{\nu}{2} {\rm log} (x) - \frac{\nu}{2} {\rm log}(n-x) - \frac{\nu}{2}{\rm log}(2 \pi)+ O\left(\frac{1}{n} \right)$

Consider the integral:

$\displaystyle  \int_0^1 p^{\nu x} (1-p)^{\nu(n-x)} = \mathbb{B}(\nu x+1, \nu(n-x)+1) = \frac{ \Gamma(\nu x+1) \Gamma(\nu(n-x)+1)}{\Gamma(\nu n +2)}$

So by Stirling formula again:

$\displaystyle U_1 =  {\rm log} \left[\frac{ \Gamma(\nu x+1) \Gamma(\nu(n-x)+1)}{\Gamma(\nu n +2)}\right] =$

$\displaystyle  =\left[\nu x {\rm log}(\nu) + \nu x {\rm log} (x) - \nu x + \frac{1}{2} {\rm log} (\nu) + \frac{1}{2} {\rm log} (x) + \frac{1}{2} {\rm log}(2 \pi) \right]$

$\displaystyle +  \left[ \nu (n-x) {\rm log} (\nu) + \nu (n-x) {\rm log} (n-x) - \nu (n-x) + \frac{1}{2} {\rm log} (\nu) + \frac{1}{2} {\rm log} (n-x) + \frac{1}{2} {\rm log} (2 \pi) \right] - $

$\displaystyle  - \left[ \nu n {\rm log}(n) + \nu n {\rm log}(n) - \nu n + \frac{1}{2} {\rm log} (\nu) + \frac{1}{2} {\rm log} (n) + \frac{1}{2} {\rm log} (2 \pi) \right] - {\rm log}(\nu) - {\rm log} (n) + O\left(\frac{1}{n} \right). $

We obtain that

$\displaystyle U_0 + U_1 = \frac{1- \nu}{2} {\rm log} (x)  +\frac{1- \nu}{2} {\rm log} (n-x) + \frac{1- \nu}{2} {\rm log} (2 \pi) - \frac{1}{2} {\rm log} (\nu) + \nu {\rm log} (n) - {\rm log}(n) - \frac{1- \nu}{2} {\rm log} (n) + O\left(\frac{1}{n} \right) = $

$\displaystyle = (1-\nu) \frac{1}{2} ( -{\rm log} (n) + {\rm log} (x) + {\rm log} (n-x) + {\rm log} (2 \pi) - 2 {\rm log} (n)) - \frac{1}{2} {\rm log} (\nu) + O\left(\frac{1}{n} \right)$

$\displaystyle = \frac{1- \nu}{2} {\rm log} \left( \frac{ 2 \pi x (n-x) } {n^3} \right) - \frac{1}{2} {\rm log} (\nu)+ O\left(\frac{1}{n} \right)$

So we have that:

\begin{equation}
 H_{\nu}(f) = \frac{1}{2} {\rm log} \left( \frac{ 2 \pi x (n-x) } {n^3} \right) - \frac{{\rm log} (\nu)}{2(1-\nu)}+ O\left(\frac{1}{n} \right),
 \end{equation}
note that it tends to Renyi differential entropy of Gaussian RV as $n \to \infty$.

Taking the limit when $\nu \to 1$ and applying L'Hopital's rule we get that:

\begin{equation} H_{\nu \to 1}(f) = \lim_{\nu \to 1}H_{\nu}(f) = \frac{1}{2} {\rm log} \left( \frac{ 2 e \pi x (n-x) } {n^3} \right) + O\left(\frac{1}{n} \right) .
 \end{equation}

For example, when $x= \alpha n$, $0<\alpha <1$ the Renyi entropy:

$$ H_{\nu \to 1}(f) = \frac{1}{2} {\rm log}\frac{2 \pi e [\alpha (1-\alpha)]}{n}+ O\left(\frac{1}{n} \right),$$
where the first term is Shannon's entropy of Gaussian RV with corresponding variance.

Or similarly  when $x= n^{\beta}$, $0<\beta < 1$ the Renyi entropy:
$$ H_{\nu \to 1}(f) = \frac{1}{2} {\rm log}\frac{2 \pi e (1-n^{\beta-1})}{n^{2-\beta}}+ O\left(\frac{1}{n^\beta} \right)$$
where the first term is Shannon's differential entropy of Gaussian RV with variance $ \displaystyle \sigma^2 = \frac{1-n^{\beta-1}}{n^{2-\beta}}$.

\textbf{(b)} 
In this case when $\phi^{(n)}(p)$ is given in (\ref{weight}) and $x=\alpha n$, the weighted Renyi entropy has the form:
$$H_{\nu}^{\psi}(f) = \frac{1}{1-\nu} {\rm log} \int_{0}^1  \phi^{(n)}(p) \left( f(p) \right)^{\nu}{\rm d}p$$

$\displaystyle \int_{0}^1  \phi^{(n)}(p) \left( f(p) \right)^{\nu}{\rm d}p =   U_1 U_2 U_3$, where

$\displaystyle  U_1 = \frac{\Gamma(\nu x + \gamma \sqrt{n}+1) \Gamma( \nu(n-x) + (1-\gamma)\sqrt{n}+1)}{\Gamma(\nu n + \sqrt{n} +2)}$
;
$\displaystyle  U_2= \left( \frac{\Gamma(n+2)}{\Gamma(x+1)\Gamma(n-x+1)} \right)^{\nu-1}$;

$\displaystyle U_3=\frac{\Gamma(n+ \sqrt{n}+2)}{\Gamma(x+z+1) \Gamma(n-x + \sqrt{n} - z+1)}$

$\displaystyle  {\rm log}(U_1) = \nu x {\rm log}(x) + z {\rm log}(x) + \frac{1}{2} {\rm log}(x) + \nu (n-x) {\rm log}(n-x) + (\sqrt{n}-z) {\rm log}(n-x) + \frac{1}{2} {\rm log}(2 \pi) -\frac{1}{2} {\rm log}(\nu) + \frac{1}{2} {\rm log}(n-x) - \nu n {\rm log}(n) - \sqrt{n} {\rm log}(n)- \frac{1}{2} {\rm log}(n) - {\rm log}(n) + \frac{(\alpha-\gamma)^2}{2\alpha (1-\alpha) \nu} +\frac{1}{2} {\rm log} \left( \frac{2 \pi \alpha (1-\alpha)}{\nu} \right)+ O\left(\frac{1}{n} \right)$

$\displaystyle  {\rm log}(U_2) = \nu n {\rm log}(n) - \nu x {\rm log}(x) - \nu (n-x) {\rm log}(n-x) + \nu {\rm log}(n) + \frac{\nu}{2}({\rm log}(n) - {\rm log}(x) - {\rm log}(n-x) - {\rm log}(2 \pi)) - n {\rm log}(n) + x{\rm log}(x) + (n-x) {\rm log}(n-x) - {\rm log}(n) - \frac{1}{2}({\rm log}(n) - {\rm log}(x) - {\rm log}(n-x) - {\rm log} (2 \pi))+ O\left(\frac{1}{n} \right)$

$\displaystyle  {\rm log}(U_3)= {\rm log}(n) + n {\rm log}(n) + \sqrt{n}{\rm log}(n) - x {\rm log}(x) - z {\rm log}(x) - (n-x) {\rm log}(n-x) - (\sqrt{n} -z) {\rm log}(n-x) + \frac{1}{2} ({\rm log}(n) - {\rm log}(x) - {\rm log}(2 \pi) - {\rm log}(n-x)) - \frac{(\alpha-\gamma)^2}{2\alpha (1-\alpha)} -\frac{1}{2} {\rm log} \left(2 \pi \alpha (1-\alpha) \right) +O\left(\frac{1}{n} \right)$

Taking all parts together, we obtain that
\begin{equation}H^{\phi}_{\nu}(f) =  \frac{1}{2} {\rm log} \frac{2 \pi \alpha (1-\alpha)}{n}  -\frac{log(\nu)}{2(1-\nu)} + \frac{(\alpha-\gamma)^2}{2\alpha (1-\alpha)(1-\nu)} \left( \frac{1}{\nu} -1 \right) + O\left(\frac{1}{n} \right) 
 \end{equation}
Taking the limit when $\nu \to 1$ and applying L'Hopital's rule we get that:
\begin{equation} H^{\phi}_{1}(f) = \lim_{\nu \to 1}H_{\nu}(f) = \frac{1}{2} {\rm log}\frac{2 \pi e [\alpha (1-\alpha)]}{n} + \frac{(\alpha-\gamma)^2}{2\alpha (1-\alpha)}+ O\left(\frac{1}{n} \right)
 \end{equation}
So the weighted Reniy entropy tends to Shannon's weighted entropy as $\nu \to 1$.
\end{proof}

\begin{proposition}
For any continuous random variable $X$ with PDF $f(x)$ and for any non-negative weight function $\phi(x)$ which satisfies condition (\ref{norm}) and such that
$$\int_{\mathbb{R}}\phi(x) (f(x))^\nu \vert {\rm log} (f(x)) \vert {\rm d}x<\infty,$$
 the weighted Renyi differential entropy $H^{\phi}_{\nu}(f)$ is a non-increasing function of $\nu$ and
\begin{equation}
\frac{\partial}{\partial \nu} H^{\phi}_{\nu}(f) = -\frac{1}{(1-\nu)^2} \int_{\mathbb{R}}z(x) {\rm log} \frac{z(x)}{\phi(x) f(x)} dx,
\end{equation}
\noindent where

$$z(x) = \frac{\phi(x) (f(x))^\nu}{\int_{\mathbb{R}}\phi(x) (f(x))^\nu {\rm d}x}$$
Similarly, the Tsallis weighted entropy $S^{\phi}_{q}(f)$ given in (\ref{tsal}) is a non-increasing function of $q$.
\end{proposition}

\begin{proof}

We need to show that

$$\frac{\partial}{\partial \nu} H^{\phi}_{\nu}(f)\leq 0.$$

\begin{equation}
 \frac{\partial}{\partial \nu} H^{\phi}_{\nu}(f) = \frac{  {\rm log} \int_{\mathbb{R}}\phi(x) (f(x))^\nu  {\rm d}x}{(1-\nu)^2} + \frac{  \int_{\mathbb{R}}\phi(x) (f(x))^\nu {\rm log} (f(x)) {\rm d}x}{(1-\nu)\int_{\mathbb{R}}\phi(x) (f(x))^\nu {\rm d}x} = I_1 + I_2
 \end{equation}

\noindent Denote
\begin{equation}
z(x) = \frac{\phi(x) (f(x))^\nu}{\int_{\mathbb{R}}\phi(x) (f(x))^\nu{\rm d}x}.
\label{subs}
 \end{equation}
Note that $z(x) \geq 0$ for any $x$ and 
$$\int_{\mathbb{R}} z(x) {\rm d}x =1 $$

Let $ \displaystyle Q_1= \int_{\mathbb{R}}\phi(x) (f(x))^\nu{\rm d}x $ and $ \displaystyle Q_2 = {\rm log} \int_{\mathbb{R}}\phi(x) (f(x))^\nu{\rm d}x$.

Using the substitution (\ref{subs})
\begin{equation}
Q_2= {\rm log} (\phi(x)) + \nu {\rm log} (f(x)) - {\rm log}(z(x)).
\label{subs2}
 \end{equation}

We have that

$\displaystyle I_2 = \frac{1}{1-\nu} \frac{ Q_1 \int_{\mathbb{R}}  z(x) {\rm log}(f(x)) {\rm d}x}{Q_1} = \frac{1}{1-\nu} \int_{\mathbb{R}}  z(x) {\rm log}(f(x)) {\rm d}x  $

$\displaystyle I_1+I_2 = \frac{1}{(1-\nu)^2} \left({\rm log} \int_{\mathbb{R}}\phi(x) (f(x))^\nu  {\rm d}x + (1-\nu)\int_{\mathbb{R}}  z(x) {\rm log}(f(x)) {\rm d}x \right)= \frac{1}{(1-\nu)^2} I_3$

By substitution ${\rm log} (f(x))$ using (\ref{subs2}) we get:

$\displaystyle I_3 = Q_2 + (1-\nu) \left( \frac{Q_2}{\nu} + \frac{1}{\nu} \int_{\mathbb{R}} z(x) {\rm log}(z(x)) {\rm d}x - \frac{1}{\nu} \int_{\mathbb{R}} z(x) {\rm log}(\phi(x)) {\rm d}x \right)=$

$\displaystyle \frac{Q_2}{\nu} + \frac{1}{\nu}\int_{\mathbb{R}} z(x) {\rm log}(z(x)) {\rm d}x - \int_{\mathbb{R}} z(x) {\rm log}(z(x)) {\rm d}x + \int_{\mathbb{R}} z(x) {\rm log}(\phi(x)) {\rm d}x - \frac{1}{\nu} \int_{\mathbb{R}} z(x) {\rm log}(\phi(x)) {\rm d}x   $

Applying (\ref{subs2}) again we get that

$\displaystyle  I_3 = \int_{\mathbb{R}} z(x) {\rm log}(f(x)) {\rm d}x  - \int_{\mathbb{R}} z(x) {\rm log}(z(x)) {\rm d}x + \int_{\mathbb{R}} z(x) {\rm log}(\phi(x)) {\rm d}x = - \int_{\mathbb{R}} z(x) {\rm log}\left(\frac{z(x)}{\phi(x) f(x)} \right) {\rm d}x  $

We obtain that

\begin{equation}
-\frac{\partial}{\partial \nu} H^{\phi}_{\nu}(f) = \frac{1}{(1-\nu)^2}  \int_{\mathbb{R}} z(x) {\rm log}\left(\frac{z(x)}{\phi(x) f(x)} \right) {\rm d}x= \frac{1}{(1-\nu)^2}  \mathbb{D}_{KL}(z||\phi f). 
\end{equation}
Here $\mathbb{D}_{KL}(z||\phi f)$ is Kullback–Leibler divergence between $z$ and $\phi f$ which is always non-negative. Due to conditions $\phi(x) f(x) \geq 0$ and (\ref{norm}), $\phi(x) f(x)$ is itself a PDF:
$$
\int_{\mathbb{R}} \phi(x) f(x) {\rm d}x = 1
$$
Similarly, one can show that Tsallis weighted differential entropy given in (\ref{tsal}) is non-increasing function of $q$. So, the result follows.
\end{proof}

\noindent \begin{theorem}
Let $Z^{(n)}$ be a RV with  $f^{(n)}$ - conditional PDF after $n$ trials given by (\ref{pdf}) with the weighted Tsallis differential entropy $S_{q}(f^{(n)})$ given in (\ref{tsal}).

\textbf{(a)} When both $(x)$ and $(n-x)$ tend to infinity as $n \to \infty$ and $\phi^{(n)}(p) = 1,$
\begin{equation}
\lim_{n \rightarrow \infty} \left(S_{q}(f^{(n)})-\frac{1}{q-1} \left(1-\frac{1}{\sqrt{q}} \left( \frac{ 2 \pi x (n-x) } {n^3} \right)^{\frac{1-q}{2}}  \right) \right) = 0.
\end{equation}
For any fixed $n$ the Tsallis differential entropy tends to Shannon's differential entropy as $q \to 1$.

\textbf{(b)} When $x=\alpha n$ and the weight function $\phi^{(n)}(p)$ given in (\ref{weight})
\begin{equation}
 \lim_{n \rightarrow \infty} \left( S^{\phi}_{q}(f_{\alpha}^{(n)}) - \frac{1}{q-1} \left(1-\frac{1}{\sqrt{q}} \left( \frac{ 2 \pi \alpha (1-\alpha) } {n} \right)^{\frac{1-q}{2}} \exp \left( \frac{(\alpha-\gamma)^2 (1-q)}{2\alpha (1-\alpha)q}  \right) \right) \right) = 0 
\end{equation}
The weighted Tsallis differential entropy tends to Shannon's weighted differential entropy RV with PDF given in (\ref{pdf}) as $q \to 1$.
\label{theorem2}
\end{theorem}
\begin{remark}
It can be seen from Theorem 4(a) and Theorem 5(a) that for large $n$ Renyi's entropy and Tsallis's entropy (for $\phi \equiv 1$) "behaves" like respective entropies of Gaussian RV with variance $\sigma^2 = \frac{x (n-x)}{n^3}$.
\end{remark}
\begin{proof}

\textbf{(a)} In this case $\phi^{(n)}(p)  \equiv 1$, the Tsallis entropy have the form:

$\displaystyle S_{q}(f) = \frac{1}{q-1} \left(1- \int_{0}^{1}  \left( f(p) \right)^{q}{\rm d}p \right)= \frac{1}{q-1} \left(1- \int_{0}^{1}  \left( (n+1) {n\choose x} p^x (1-p)^{n-x} \right)^{q}{\rm d}p \right)  $

It was shown above that

$$\displaystyle {\rm log} \int_0^1 \left( f(p) \right)^{q}{\rm d}p \simeq \frac{1- q}{2} {\rm log} \left( \frac{ 2 \pi x (n-x) } {n^3} \right) - \frac{1}{2} {\rm log} (q)$$
So we have that 
$$V_0=\int_0^1 \left( f(p) \right)^{q}{\rm d}p \simeq  \frac{1}{\sqrt{q}} \left( \frac{ 2 \pi x (n-x) } {n^3} \right)^{\frac{1-q}{2}} $$
We straightforwardly obtain that
\begin{equation}
\displaystyle S_{q}(f) \simeq \frac{1}{q-1} \left(1-\frac{1}{\sqrt{q}} \left( \frac{ 2 \pi x (n-x) } {n^3} \right)^{\frac{1-q}{2}}  \right)
\end{equation}
Note that $V_0 \to 1$ when $ q \to 1$, applying L'Hospital's rule we get that:
\begin{equation}
\lim_{q \to 1} S_{q}(f) = S_{1}(f) \simeq \frac{1}{2} {\rm log} \left( \frac{ 2 e \pi x (n-x) } {n^3} \right) 
\end{equation}

The first term in expression above is nothing else but Shannon's differential entropy of Gaussian RV.

\textbf{(b)} In this case when $\phi^{(n)}(p)$ is given in (\ref{weight}) the Tsallis entropy have the form:
$$S^{\phi}_{q}(f) = \frac{1}{q-1} \left(1- \int_{0}^1  \phi^{(n)}(p) \left( f(p) \right)^{q}{\rm d}p \right)$$
Using that $x=\alpha n$ and by Stirling's formula, it was shown above that
$$\displaystyle {\rm log} \left[\int_{0}^1  \phi^{(n)}(p) \left( f(p) \right)^{q}{\rm d}p \right] \simeq \frac{1-q}{2} {\rm log} \frac{2 \pi \alpha (1-\alpha)}{n}  -\frac{{\rm log}(q)}{2} + \frac{(\alpha-\gamma)^2}{2\alpha (1-\alpha)} \left( \frac{1}{q} -1 \right)$$
So we have:
$$V_1=\int_{0}^1  \phi^{(n)}(p) \left( f(p) \right)^{q}{\rm d}p \simeq \frac{1}{\sqrt{q}} \left( \frac{ 2 \pi \alpha (1-\alpha) } {n} \right)^{\frac{1-q}{2}} \exp \left( \frac{(\alpha-\gamma)^2}{2\alpha (1-\alpha)} \left( \frac{1}{q} -1 \right) \right)$$
Weighted Tsallis entropy:
\begin{equation}
S^{\phi}_{q}(f(p)) \simeq  \frac{1}{q-1} \left(1-\frac{1}{\sqrt{q}} \left( \frac{ 2 \pi \alpha (1-\alpha) } {n} \right)^{\frac{1-q}{2}} \exp \left( \frac{(\alpha-\gamma)^2}{2\alpha (1-\alpha)} \left( \frac{1}{q} -1 \right) \right) \right)
\end{equation}
Note that $V_0 \to 1$ when $ q \to 1$, applying L'Hospital's rule we get that:
\begin{equation}
 S^{\phi}_{1}(f) = \lim_{q \to 1} S^{\phi}_{q}(f) \simeq  \frac{1}{2} {\rm {\rm log}}\frac{2 \pi e [\alpha (1-\alpha)]}{n} + \frac{(\alpha-\gamma)^2}{2\alpha (1-\alpha)}.
 \end{equation}
Then the weighted Tsallis entropy tends to weighted Shannon's differential entropy when $q \to 1$. 
\end{proof}
\begin{theorem}
Let $Z^{(n)}$ be a RV with  $f_{\alpha}^{(n)}$ - conditional PDF after $n$ trials given by (\ref{pdf}), when $x = \alpha n$ {\rm (}$0<\alpha<1${\rm )} and $I(f_{\alpha}^{(n)})$ is the weighted Fisher information of $Z^{(n)}$ given in (1.5):

\textbf{(a)} When $\phi^{(n)}(p) = 1,$
\begin{equation}
\lim_{n \rightarrow \infty} \left[I(f_{\alpha}^{(n)})- \left(\frac{1}{\alpha(1-\alpha)} \right) n \right] = - \frac{2\alpha^2 - 2 \alpha +1}{2\alpha^2 (1-\alpha)^2}.
\end{equation}
\label{theorem3}

\textbf{(b)} When $\phi^{(n)}(p)$ is given in (\ref{weight}):
\begin{equation}
 \lim_{n \rightarrow \infty} \left[I^{\phi}(f_{\alpha}^{(n)})- \left(\frac{1}{\alpha(1-\alpha)} + \frac{(\alpha-\gamma)^2}{(1-\alpha)^2 \alpha^2} \right) n - B(\alpha, \gamma)\sqrt{n} \right] = C(\alpha,\gamma),
\end{equation}
where $B(\alpha, \gamma)$ and $C(\alpha,\gamma)$ are constants which depend only on $\alpha$ and $\gamma$ and are given in (\ref{constantB}) and (\ref{constantC}) respectively .

\end{theorem}

\begin{proof}
\textbf{(a)} The Fisher information in the case $\phi^{(n)}(p) = 1$ and $x=\alpha n$  takes the form:

$$I(\alpha) = \mathbb{E} \left( \left( \frac{{\rm \partial}}{{\rm \partial} \alpha} {\rm log} f(p;\alpha) \right)^2 { \Big | \alpha }  \right)  = \int_0^1 \left( \frac{{\rm \partial}}{{\rm \partial} \alpha} {\rm log} f(p;\alpha) \right)^2 f(p, \alpha) dp,$$

where $f=f_{\alpha}^{(n)}$.  Next,
$$ {\rm log} ( f(p, \alpha)) = \alpha n {\rm log} (p) + (1-\alpha)n {\rm log} (1-p) + {\rm log}(n+1)! - {\rm log}(x!) - {\rm log} ((n-x)!)$$
and
\begin{equation}
 \frac{{\rm \partial}}{{\rm \partial} \alpha} {\rm log} f(p;\alpha) = n {\rm log} (p) - n {\rm log} (1-p) + n \psi(n-x+1) - n \psi(x+1),
 \end{equation}
$\displaystyle \left( \frac{{\rm \partial}}{{\rm \partial}} {\rm log} f(p;\alpha) \right)^2 = n^2 {\rm log}^2(p) + n^2 {\rm log}^2 (1-p) + n^2 \psi^2 (n-x+1) + n^2 \psi^2(x+1) - 2n^2{\rm log}(p) {\rm log} (1-p) + 2n^2 {\rm log}(p) \psi (n-x+1) - 2 n^2 {\rm log}(p) \psi (x+1) - 2 n^2 {\rm log} (1-p) \psi (n-x+1) +2n^2 {\rm log}(1-p) \psi (x+1) - 2n^2 \psi(x+1) \psi(n-x+1)$.

For the following computation of expectation we will need so following integrals:

$\displaystyle  \int_0^1 ({\rm log}(p))^2 p^x (1-p)^{n-x} dp = \frac{\Gamma(n-x+1) \Gamma (x+1)}{\Gamma(n+2)} ( \psi(n+2) - \psi(x+1))^2 - \psi^{(1)}(n+2) + \psi^{(1)}(x+1) $,
where $\Gamma(x)$ is a Gamma function and  $\psi^{(1)}(x)$ is the first derivative of digamma function.

$\displaystyle  \int_0^1 ({\rm log}(1-p))^2 p^x (1-p)^{n-x} dp = \frac{\Gamma(n-x+1) \Gamma (x+1)}{\Gamma(n+2)}( \psi(n+2) - \psi(n-x+1))^2 - \psi^{(1)}(n+2) + \psi^{(1)}(n-x+1)$

$\displaystyle  \int_0^1 {\rm log}(p){\rm log}(1-p) p^x (1-p)^{n-x} dp = \frac{\Gamma(n-x+1) \Gamma (x+1)}{\Gamma(n+2)}( \psi(n+2) - \psi(n-x+1)(\psi(n+2) - \psi(x+1)) - \psi^{(1)}(n+2)$

$\displaystyle  \int_0^1 {\rm log}(p) p^x (1-p)^{n-x} dp = \frac{\Gamma(n-x+1) \Gamma (x+1)}{\Gamma(n+2)}(-\psi(n+2) + \psi(x+1))$

$\displaystyle  \int_0^1 {\rm log}(1-p) p^x (1-p)^{n-x} dp = \frac{\Gamma(n-x+1) \Gamma (x+1)}{\Gamma(n+2)}(-\psi(n+2) + \psi(n-x+1)$

So, we have that

$\displaystyle \int_0^1 \left( \frac{{\rm \partial}}{{\rm \partial}\alpha} {\rm {\rm log}} f(p;\alpha) \right)^2 f(p, \alpha) dp = $

$\displaystyle n^2 (n+1) {n\choose x} \frac{\Gamma(n-x+1) \Gamma (x+1)}{\Gamma(n+2)} ( (\psi(n+2))^2 + (\psi(x+1))^2 - 2 \psi(n+2)\psi(x+1) - \psi^{(1)}(n+2)+  \psi^{(1)}(x+1) + (\psi(n+2))^2 + (\psi(n-x+1))^2 - 2 \psi(n+2) \psi(n-x+1) - \psi^{(1)}(n+2) + \psi^{(1)} (n-x+1) + (\psi(n-x+1))^2 + (\psi(x+1))^2 - 2(\psi(n+2))^2 + 2\psi(n+2) \psi(x+1) + 2 \psi(n-x+1) \psi(n+2) - 2 \psi(n-x+1) \psi(x+1) + 2\psi^{(1)} (n+2) - 2\psi(n-x+1)\psi(n+2) + 2 \psi(n-x+1)\psi(x+1) + 2\psi(x+1)\psi(n+2) - 2(\psi(x+1))^2 - 2(\psi(n-x+1))^2 + 2\psi(n-x+1)\psi(n+2) + \psi(x+1)\psi(n-x+1) -2 \psi(x+1)\psi(n+2) - 2\psi(x+1)\psi(n-x+1)) =  $

$\displaystyle = n^2 (  \psi^{(1)}(x+1) + \psi^{(1)} (n-x+1))$

\begin{equation}
I(\alpha) = n^2 (  \psi^{(1)}(x+1) + \psi^{(1)} (n-x+1)).
\end{equation}

Using the asymptotic for the digamma function we can rewrite:

\begin{equation}
I(\alpha)= \frac{1}{\alpha(1-\alpha)}n - \frac{1}{2} \frac{2\alpha^2 - 2 \alpha +1}{\alpha^2 (1-\alpha)^2} + O\left(\frac{1}{n} \right).
\end{equation}

\begin{remark}
When $x = \alpha n$
$$ \int_0^1 p f^{(n)}_{\alpha} {\rm d} p = \alpha + b_n(\alpha),$$
where $b_n(\alpha)$ is a bias.
$$b_n(\alpha) \simeq \frac{1-2\alpha}{n}$$
Note that $$ \frac{\partial}{\partial \alpha} b_n(\alpha) \simeq - \frac{2}{n} \to 0 $$ as $n \to \infty$. So, our estimate is asymptotically unbiased. Also note that the first term in Theorem \ref{theorem3} has the same form as in the classical problem of estimating $p$ in a series of binary trials $\frac{n}{p (1-p)}$.
\end{remark}
\textbf{(b)}
The weighted Fisher Information in the case $x=\alpha n$ ($0<\alpha<$) takes the following form:
$$I^{\phi}(f) = \mathbb{E} \left( \phi^{(n)}(p) \left( \frac{{\rm \partial}}{{\rm \partial} \alpha} {\rm {\rm log}} f(p;\alpha) \right)^2 { \Big | \alpha }  \right)  = \int_0^1 \phi^{(n)}(p) \left( \frac{{\rm \partial}}{{\rm \partial}\alpha} {\rm {\rm log}} f(p;\alpha) \right)^2 f(p, \alpha) dp$$
where the $\phi^{(n)}(p)$ is given in (\ref{weight}).

The second term under integral $\displaystyle \left( \frac{{\rm \partial}}{{\rm \partial}\alpha} {\rm log} f(p;\alpha) \right)^2$ can be found as before exactly.

Let $\displaystyle W=\frac{\Gamma(n-x+1+\sqrt{n}-z) \Gamma (x+1+z)}{\Gamma(n+2+\sqrt{n})}$.
So in order to compute the weighted Fisher information we will need to compute following integrals.

$\displaystyle  \int_0^1 ({\rm log}(p))^2 p^{z+x} (1-p)^{n-x+\sqrt{n}-z} dp = W( \psi(n+2+\sqrt{n}) - \psi(x+z+1))^2 - \psi^{(1)}(n+2+\sqrt{n}) + \psi^{(1)}(x+z+1)$

$\displaystyle  \int_0^1 ({\rm log}(1-p))^2 p^{z+x} (1-p)^{n-x+\sqrt{n}-z} dp = W( \psi(n+2+\sqrt{n}) - \psi(n-x+1+\sqrt{n}-z))^2 - \psi^{(1)}(n+2+\sqrt{n}) + \psi^{(1)}(n-x+1+\sqrt{n}-z)$

$\displaystyle  \int_0^1 {\rm log}(p){\rm log}(1-p) p^{z+x} (1-p)^{n-x+\sqrt{n}-z} dp = W( \psi(n+2+\sqrt{n}) - \psi(n-x+1+\sqrt{n}-z)(\psi(n+2+\sqrt{n}) - \psi(x+1+z)) - \psi^{(1)}(n+2+\sqrt{n})$

$\displaystyle  \int_0^1 {\rm log}(p) p^{z+x} (1-p)^{n-x+\sqrt{n}-z} dp = W(-\psi(n+2+\sqrt{n}) + \psi(x+1+z)$

$\displaystyle  \int_0^1 {\rm log}(1-p) p^{z+x} (1-p)^{n-x+\sqrt{n}-z} dp = W(-\psi(n+2+\sqrt{n}) + \psi(n-x+1+\sqrt{n}-z)$
 
 Taking all parts together:

$ \displaystyle I^{\phi}(f^{(n)}_{\alpha})= n^2\left(\psi^{(1)}(x+ z+1) + \psi^{(1)}(n-x+1 + \sqrt{n} - z)\right)+$

$\displaystyle + n^2 \left[ \left(\psi (x+z+1) - \psi(x+1)\right)^2 + \left(\psi(n-x+1 + \sqrt{n} -z) - \psi(n-x+1)\right)^2\right]+$

$\displaystyle + 2 n^2 \left[  \left(\psi(n-x+1) - \psi(n-x + \sqrt{n} - z+1)\right)    \left(\psi(x+z+1) - \psi(x+1)\right) \right]$

Using the asymptotic for the digamma function we can rewrite:

\begin{equation}
I(\alpha) = A(\alpha,\gamma)n + B(\alpha, \gamma)\sqrt{n} + C(\alpha,\gamma) + O\left( \frac{1}{\sqrt{n}} \right)
\end{equation}

where

\begin{equation}
 A(\alpha,\gamma)= \frac{1}{\alpha(1-\alpha)} + \frac{(\alpha-\gamma)^2}{(1-\alpha)^2 \alpha^2}
\end{equation}
\begin{equation}
 B(\alpha, \gamma)= \frac{ 2\alpha \gamma - \gamma -\alpha^2}{(1-\alpha)^2\alpha^2}+ \frac{(\alpha-\gamma)^2}{(1-\alpha)^3 \alpha^3}( \alpha(2\gamma-1)-\gamma) 
 \label{constantB}
 \end{equation}
\begin{equation}
\begin{array}{l} \displaystyle C(\alpha,\gamma)= \frac{\alpha - 2\alpha^4 - 2\gamma^2 + 6\alpha \gamma^3 + \alpha^3(2+4\gamma) - 3\alpha^(1+\gamma^2)}{-2(1-\alpha)^3 \alpha^3)}+\\ \displaystyle + \frac{\alpha^4(-31 - 44\gamma + 72 \gamma^2-56\gamma^3 + 28 \gamma^4 + 36\alpha  - 12 \alpha^2)}{12 (1-\alpha)^4 \alpha^4}+ \\ \displaystyle + \frac{6 \alpha^2(\gamma^2-2\gamma^3+12\gamma^4-1) - 4\gamma^3(11\gamma - 44\alpha \gamma-6+3\gamma^2-6\gamma^3+14 \gamma^4)}{12 (1-\alpha)^4 \alpha^4} \end{array}
\label{constantC}
\end{equation}

A r$\hat{{\rm o}}$le of the weight function of form (\ref{weight}) results in appearance of the term of order $\sqrt{n}$, but the main order, $n$, remains the same. However, the coefficient in front of it is higher by $ \displaystyle \frac{(\alpha-\gamma)^2}{(1-\alpha)^2 \alpha^2}$.Evidently, when the frequency  of special interest is equal to the true frequency the leading term is the same as in Fisher Information with constant weight. Also note that the rate depends on the distance between $\gamma$ and $\alpha$ and when $\gamma \to \alpha$ the only first terms remains.
\end{proof}

\section*{Weighted inequalities}
Recall that $\bZ_\theta \in \mathbb{R}^d$ is the family of RV with PDF $f_\theta$ where $\theta \in \Theta \subset \mathbb{R}^m$ is the vector of parameters of PDF $f_\theta$. Let $\phi(\bz,\theta,\gamma)$ be the continuous positive weight function defined in (\ref{weight0}), $ {\rm I}^{\phi}(\theta)$ be the weighted Fisher information $(m \times m)$ matrix given in (\ref{fisher}) and $g(\theta)$ be the weighted expectation given in (\ref{expect}). Let $\mathbb{V}^{\phi}_\theta(\bZ)$ be the weighted covariance matrix of RV $\bZ_\theta$
\begin{equation}
\mathbb{V}^{\phi}_\theta(\bZ) = \mathbb{E}^\phi_\theta \left[(\bZ-e(\theta))(\bZ-e(
\theta))^{\rT} \right].
\label{var}
\end{equation}
We also assume that in (\ref{expect}) and (\ref{norm}) differentiation with respect to the parameters up to order to be considered under the sign of the integration is valid. So, the equality (\ref{integral2}) (and analogous) holds. A sufficient condition for this is that the integrand after the operation of differentiation $\eta(\theta)$ is bounded by an integrable function $\chi$ which does not depend on $\theta$ 
$$|\eta(\theta)| \leq \chi,$$
i.e. the integral converges uniformly in $\theta$. 

In the following sections we consider the special class of weight functions which can be represented in the following form:
\begin{equation}
\phi(\bz,\theta,\gamma) = \frac{1}{\kappa(\theta,\gamma)}\tilde{\phi}(\bz,\gamma).
\label{weight0}
\end{equation}
Here $\kappa(\theta,\gamma) \in {\rm \textbf{C}}^{k}$ where ${\rm \textbf{C}}^k$ is the family of function  with continuous derivatives up to order $k$ ($k$ will be specified below), and $\kappa(\theta,\gamma)$ is found from the normalizing condition
\begin{equation}
\int_{\mathbb{R}^d} \phi(\bz,\theta,\gamma) f(\bz){\rm d}\bz = 1
\label{norm}
\end{equation}
as before. Note that the condition (\ref{norm}) can be rewritten in the following form
\begin{equation}
\int_{\mathbb{R}^d} \tilde{\phi}(\bz,\gamma) f(\theta,\bz) {\rm d}\bz =\kappa(\theta,\gamma)
\label{norm2}
\end{equation}
where $\tilde{\phi}(\bz,\gamma)$ is a function that have a sharp peak at the point $\gamma$ and does not depend on $\theta$.

In the Bayesian framework we consider RV $Z_\alpha^{(n)}$ with a PDF $f^{(n)}=f^{(n)}_\alpha$ given in (\ref{pdf}) assuming that $x=x(n)=\lfloor \alpha n \rfloor$, considered in the Bayesian problem stated above \cite{Mozg2015_1}. The explicit asymptotic  expansions for lower bound are obtained in cases of the following weight functions:
\begin{equation}
\phi_1^{(n)}(p) = \frac{1}{\kappa_1(\alpha,\gamma)} p^{\gamma} (1-p)^{1-\gamma},
\label{weight01}
\end{equation}
\begin{equation}
\phi_2^{(n)}(p) = \frac{1}{\kappa_2(\alpha,\gamma)} p^{\gamma \sqrt{n}} (1-p)^{(1-\gamma)\sqrt{n}},
\label{weight02}
\end{equation}
\begin{equation}
\phi_3^{(n)}(p) = \frac{1}{\kappa_3(\alpha,\gamma)} p^{\gamma n} (1-p)^{(1-\gamma)n}
\label{weight03}
\end{equation}
where ${\kappa_i(\alpha,\gamma)}$, $i=1,2,3$ are found from the condition (\ref{norm}).

Denote the partial derivative of order $j$
$$f^{(j)} = \frac{ \partial^{j} f}{\partial \theta^j}.$$

Recall $\psi^{(0)}(x) = \psi(x)$ and by $\psi^{(1)}(x)$ the digamma function and its first derivative respectively
\begin{equation}  \psi^{(n)}(x) = \frac{{\rm d^{n+1}}}{{\rm d}x^{n+1}} {\rm log} 
\left( \Gamma(x) \right)
\end{equation}
where $\Gamma(x)$ is the Gamma-function.
In further calculations the asymptotic of these functions  for $x \to \infty$ will be used \cite[\#8.362.2]{rg2007}
\begin{equation} \displaystyle \psi(x) = {\rm log}(x) - \frac{1}{2x} + O\left( \frac{1}{x^2} \right) \ {\rm as} \ x \to \infty,
\label{digamma1}
\end{equation}
\begin{equation} \displaystyle \psi^{(1)}(x) = \frac{1}{x} + \frac{1}{2x^2} + O \left( \frac{1}{x^3} \right) \ {\rm as} \ x \to \infty.
\label{digamma2}
\end{equation}

\section{Weighted Rao-Cram\'er inequality}
\begin{theorem}\textbf{(Weighted Rao-Cram\'er inequality).} Assume that 
\begin{equation}
\frac{\partial g (\theta)}{\partial \theta} = \int_{\mathbb{R}^d} \bz \frac{\partial }{\partial \theta} \left[ f_\theta(\bz) \phi(\bz,\theta,\gamma) \right] {\rm d}\bz. 
\label{integral2}
\end{equation}
Note that (\ref{integral2}) holds if integral in its RHS converges uniformly in $\theta$.
 Then the following inequality for weighted covariance matrix $\mathbb{V}^{\phi}_\theta(\bZ)$ holds
\begin{equation}
\mathbb{V}^{\phi}_\theta(\bZ) \geq \left(\frac{\partial g(\theta)}{\partial \theta} - \frac{\kappa^{\prime}(\theta,\gamma)}{\kappa(\theta,\gamma)}\left(e(\theta) - g(\theta)\right)\right){\rm I}^{\phi}(\theta)^{-1}\left(\frac{\partial g(\theta)}{\partial \theta} - \frac{\kappa^{\prime}(\theta,\gamma)}{\kappa(\theta,\gamma)}\left(e(\theta) - g(\theta)\right)\right)^{\rT}.
\label{rao}
\end{equation}
\end{theorem}
\begin{proof}
Consider the following integral
\begin{equation}
 g(\theta) \equiv \int_{\mathbb{R}^d} \bz \phi(\bz,\theta,\gamma) f_\theta(\theta,\bz) {\rm d} \bz.
 \label{integral}
 \end{equation}
Differentiating both sides in (\ref{integral}) and in (\ref{norm2}) with respect to $\theta$ and multiplying the latter one by $e(\theta)$ defined in (\ref{expectclasic}) 
\begin{equation}\int_{0}^{1} \bz \phi(\bz,\theta,\gamma)  \frac{\partial f_\theta}{\partial \theta} {\rm d} \bz -  \frac{\kappa^{\prime}(\theta,\gamma)}{\kappa^2(\theta,\gamma)} \int_0^1 \bz {\tilde \phi}(\bz, \gamma) f_\theta(\theta,\bz) {\rm d}\bz=  \frac{\partial g(\theta)}{\partial \theta},
\label{inte1}\end{equation}
\begin{equation}e(\theta) \int_{0}^{1} \phi(\bz,\theta,\gamma) \frac{\partial f_\theta}{\partial \theta} {\rm d} \bz  = \frac{\kappa^{\prime}(\theta,\gamma)}{\kappa(\theta,\gamma)}e(\theta).
\label{inte2}
\end{equation}
Subtracting (\ref{inte1}) from (\ref{inte2}),
$$\int_{0}^{1} (\bz-e(\theta)) \phi(\bz,\theta,\gamma)  \frac{\partial f_\theta}{\partial \theta} dp =\frac{\partial g(\theta)}{\partial \theta} - \frac{\kappa^{\prime}(\theta,\gamma)}{\kappa(\theta,\gamma)}\left(e(\theta) - g(\theta)\right).$$
Multiplying and dividing by $\sqrt{f_\theta}$, multiplying by conjugate vector and applying Cauchy-Schwarz inequality we get 
\begin{equation}
\mathbb{V}^{\phi}_\theta(\bZ) \geq \left(\frac{\partial g(\theta)}{\partial \theta} - \frac{\kappa^{\prime}(\theta,\gamma)}{\kappa(\theta,\gamma)}\left(e(\theta) - g(\theta)\right)\right){\rm I}^{\phi}(\theta)^{-1}\left(\frac{\partial g(\theta)}{\partial \theta} - \frac{\kappa^{\prime}(\theta,\gamma)}{\kappa(\theta,\gamma)}\left(e(\theta) - g(\theta)\right)\right)^{\rT}
\end{equation}
where $I^{\phi}=(\theta)$ is the ($m \times m$) Fisher Information matrix defined in (\ref{fisher}).

\end{proof}

\begin{theorem}
Let $Z_\alpha^{(n)}$ be a RV with a PDF $f^{(n)}_\alpha$ given in (\ref{pdf}) assuming that $x= \lfloor \alpha n \rfloor$ where $0<\alpha<1$. Then

\noindent \textbf{(a)} When weight function  $\phi(p)=\phi_1(p)$ is  given in (\ref{weight01})
\begin{equation}
\mathbb{V}^{\phi_1}(Z_\alpha) \geq \frac{\alpha (1-\alpha)}{n} + \frac{1-14\alpha + 18 \alpha^2 + 2 \gamma -8\alpha \gamma +2 \gamma^2}{2n^2} + O \left( \frac{1}{n^{5/2}} \right).
\end{equation}
\textbf{(b)} When weight function $\phi(p)=\phi_2(p)$ is  given in (\ref{weight02})
\begin{equation}
\mathbb{V}^{\phi_2}(Z_\alpha) \geq \frac{\alpha (1-\alpha) + (\alpha - \gamma)^2}{n} + \frac{-2 \alpha + \alpha^2 + \gamma + 2 \alpha \gamma - 2 \gamma^2}{n^{3/2}} + O \left( \frac{1}{n^{2}} \right).
\end{equation}
\textbf{(c)} When weight function $\phi(p)=\phi_3(p)$ is  given in (\ref{weight03})
\begin{equation}
\mathbb{V}^{\phi_3}(Z_\alpha) \geq \frac{(\alpha - \gamma)^2}{4} + C_3(\alpha,\gamma)\frac{1}{n} + O\left(\frac{1}{n^{3/2}} \right)
\end{equation}
where $C_3$ is a constant which depends only on $\alpha$ and $\gamma$ and given explicitly in (\ref{constant3}).
\end{theorem}

\begin{proof}
\textbf{\textit{(a)}}
Consider the weight function
\begin{equation}
\phi_1^{(n)}(p) = \frac{1}{\kappa_1(\alpha,\gamma)} p^{\gamma} (1-p)^{1-\gamma}
\label{weight1}
\end{equation}
where $\kappa_1(\alpha,\gamma)$ is found from the normalizing condition (\ref{norm}). Thus,
$$ \frac{1}{\kappa_1(\alpha,\gamma)}=\frac{\Gamma{(x+1)}\Gamma{(n-x+1)}\Gamma{(n+3)}}{\Gamma{(x+\gamma+1)}\Gamma{(n-x+2 - \gamma)} \Gamma{(n+2)}}.
$$
Note that the normalizing constant depends on $n$, but the remainder does not contain $n$ and $\alpha$.
For a given weight function (\ref{weight1}) the Fisher information equals:
\begin{equation}
 \begin{array}{l} 
\displaystyle I^{\phi_1}(f^{(n)}_{\alpha})= n^2\left(\psi^{(1)}(x+ \gamma +1) + \psi^{(1)}(n-x+1 +1 - \gamma)\right)+\\ \displaystyle + n^2 \left[ \left(\psi (x+\gamma+1) - \psi(x+1)\right)^2 + \left(\psi(n-x+1 +1 -\gamma) - \psi(n-x+1)\right)^2\right]+ \\ \displaystyle  2 n^2 \left[  \left(\psi(n-x+1) - \psi(n-x + 1 - \gamma +1)\right)    \left(\psi(x+\gamma+1) - \psi(x+1)\right) \right].
\end{array}
\label{information1}
\end{equation}
For the weight function (\ref{weight1}), integral in  (\ref{integral}) can be found explicitly
\begin{equation}
 \begin{array}{c}  \displaystyle \int_{0}^{1} p \phi_1^{(n)} f^{(n)}_\alpha {\rm d} p  = \frac{\Gamma(n+3)}{\Gamma(x+\gamma + 1) \Gamma(n-x-\gamma +2)} \int_0^1 p^{x+\gamma+1}(1-p)^{n-x+1-\gamma} {\rm d} p\\  \\ = \displaystyle  \frac{\Gamma(n+3) \Gamma(x+\gamma + 2)}{\Gamma(n+4) \Gamma(x+\gamma + 1)} =g_1(\alpha). \end{array}
 \label{gfunction}
 \end{equation}
 Then 
 \begin{equation}
 \frac{\partial g_1(\alpha)}{\partial \alpha} =  n\frac{\Gamma(n+3) \Gamma(x+\gamma + 2)}{\Gamma(n+4) \Gamma(x+\gamma + 1)}\left(\psi(x+\gamma+2) - \psi(x+\gamma+1) \right).
 \label{gder}
 \end{equation}
 Differentiating $\kappa(\alpha,\gamma)$ we obtain that:
 \begin{equation}
 \frac{\kappa_1^{\prime}(\alpha,\gamma)}{\kappa_1(\alpha,\gamma)} = n \left( \psi(n-x+1) - \psi(n-x+1-\gamma+1) + \psi(x+\gamma + 1) - \psi(x+1) \right).
 \label{kap1}
 \end{equation}
 Also  \begin{equation}
e(\alpha) = \frac{\Gamma(n+2)}{\Gamma(x+ 1) \Gamma(n-x+1)} \int_0^1 p^{x+1}(1-p)^{n-x} {\rm d} p= \frac{\Gamma(n+2) \Gamma(x + 2)}{\Gamma(n+3) \Gamma(x + 1)}.
 \label{pstar}
 \end{equation}
Plugging in (\ref{information1}),(\ref{gfunction}),(\ref{gder}),(\ref{kap1}) and (\ref{pstar}) in (\ref{rao}) we get 
\begin{equation}
\mathbb{V}^{\phi_1}(Z^{(n)}_\alpha) \geq \frac{\alpha (1-\alpha)}{n} + \frac{1-14\alpha + 18 \alpha^2 + 2 \gamma -8\alpha \gamma +2 \gamma^2}{2n^2} + O \left( \frac{1}{n^{5/2}} \right).
\end{equation}
\textbf{\textit{(b)}}
Consider the weight function
\begin{equation}
\phi_2^{(n)}(p) = \frac{1}{\kappa_2(\alpha,\gamma)} p^{\gamma \sqrt{n}} (1-p)^{(1-\gamma)\sqrt{n}}
\label{weight2}
\end{equation}
where $\kappa_2(\alpha,\gamma)$ is found from the normalizing condition (\ref{norm}),
$$\frac{1}{\kappa_2(\alpha,\gamma)}=\frac{\Gamma{(x+1)}\Gamma{(n-x+1)}\Gamma{(n+2+\sqrt{n})}}{\Gamma{(x+\gamma \sqrt{n}+1)}\Gamma{(n-x+1+ \sqrt{n} - \gamma \sqrt{n})} \Gamma{(n+2)}}.
$$
Note that the normalizing constant depends on $n$ as well as the remainder.
For a given weight function (\ref{weight2}) the Fisher information equals:
\begin{equation}
 \begin{array}{l} 
\displaystyle I^{\phi_2}(f^{(n)}_{\alpha})= n^2\left(\psi^{(1)}(x+ z+1) + \psi^{(1)}(n-x+1 + \sqrt{n} - z)\right)+\\ \displaystyle + n^2 \left[ \left(\psi (x+z+1) - \psi(x+1)\right)^2 + \left(\psi(n-x+1 + \sqrt{n} -z) - \psi(n-x+1)\right)^2\right]+ \\ \displaystyle  2 n^2 \left[  \left(\psi(n-x+1) - \psi(n-x + \sqrt{n} - z+1)\right)    \left(\psi(x+z+1) - \psi(x+1)\right) \right]
\end{array}
\label{information2}
\end{equation}
where $z=\gamma \sqrt{n}$.

For the weight function (\ref{weight2}), integral in  (\ref{integral}) equals
\begin{equation}
 \begin{array}{c}  \displaystyle \int_{0}^{1} p \phi_2^{(n)} f^{(n)}_\alpha{\rm d} p  = \displaystyle  \frac{\Gamma(n+\sqrt{n}+2) \Gamma(x+\gamma\sqrt{n} + 2)}{\Gamma(n+\sqrt{n}+3) \Gamma(x+\gamma\sqrt{n} + 1)} =g_2(\alpha). \end{array}
 \label{gfunction2}
 \end{equation}
 Then 
 \begin{equation}
 \frac{\partial g_2(\alpha)}{\partial \alpha} =  n\frac{\Gamma(n+\sqrt{n}+2) \Gamma(x+\gamma\sqrt{n} + 2)}{\Gamma(n+\sqrt{n}+3) \Gamma(x+\gamma\sqrt{n} + 1)}\left(\psi(x+\gamma\sqrt{n}+2) - \psi(x+\gamma\sqrt{n}+1) \right).
 \label{gder2}
 \end{equation}
 Differentiating $\kappa_2(\alpha,\gamma)$ we obtain
 \begin{equation}
 \frac{\kappa_2^{\prime}(\alpha,\gamma)}{\kappa_2(\alpha,\gamma)} = n \left( \psi(n-x+1) - \psi(n-x+1-\gamma\sqrt{n}+1) + \psi(x+\gamma\sqrt{n} + 1) - \psi(x+1) \right).
 \label{kap2}
 \end{equation}
Plugging in (\ref{information2}),(\ref{gfunction2}),(\ref{gder2}),(\ref{kap2}) and (\ref{pstar}) in (\ref{rao}) we get
\begin{equation}
\mathbb{V}^{\phi_2}(Z^{(n)}_\alpha) \geq \frac{\alpha (1-\alpha) + (\alpha - \gamma)^2}{n} + \frac{-2 \alpha + \alpha^2 + \gamma + 2 \alpha \gamma - 2 \gamma^2}{n^{3/2}} + O \left( \frac{1}{n^{2}} \right).
\label{varw2}
\end{equation}

\textbf{\textit{(c)}}
Consider the weight function
\begin{equation}
\phi_3^{(n)}(p) = \frac{1}{\kappa_3(\alpha,\gamma)} p^{\gamma n} (1-p)^{(1-\gamma)n}
\label{weight3}
\end{equation}
where $\kappa_3(\alpha,\gamma)$ is found from the normalizing condition (\ref{norm}):
$$ \frac{1}{\kappa_3(\alpha,\gamma)}=\frac{\Gamma{(x+1)}\Gamma{(n-x+1)}\Gamma{(2n+2)}}{\Gamma{(x+\gamma n+1)}\Gamma{(2n-x+1- \gamma n)} \Gamma{(n+2)}}.
$$
Note that the normalizing constant depends on $n$ as well as the remainder.
Let $y=\gamma n$ then the Fisher Information in this case equals:
\begin{equation}
 \begin{array}{l} 
\displaystyle I^{\phi_3}(f^{(n)}_{\alpha})= n^2\left(\psi^{(1)}(x+ y+1) + \psi^{(1)}(2n-x+1  - y)\right)+\\ \displaystyle + n^2 \left[ \left(\psi (x+y+1) - \psi(x+1)\right)^2 + \left(\psi(2n-x+1 -y) - \psi(n-x+1)\right)^2\right]+ \\ \displaystyle  2 n^2 \left[  \left(\psi(n-x+1) - \psi(2n-x - y+1)\right)    \left(\psi(x+y+1) - \psi(x+1)\right) \right].
\end{array}
\label{information3}
\end{equation}
Note that unlike two cases above the differences in brackets do not tend to zero, i.e.,
$$\psi(x+y+1) - \psi(x+1) = {\rm log} \left( \frac{\alpha + \gamma}{\alpha} \right) - \frac{\gamma}{2\alpha(\alpha + \gamma) n} + O \left( \frac{1}{n^2} \right).$$
Using (\ref{digamma1}) and (\ref{digamma2}), we obtain
\begin{equation}
 \begin{array}{l} 
\displaystyle I^{\phi_3}(f^{(n)}_{\alpha}) = \left( {\rm log} \frac{(1-\alpha)(\alpha+\gamma)}{\alpha(2-\alpha-\gamma)} \right)^2 n^2 + C_1(\alpha,\gamma)n + C_2(\alpha,\gamma) + O \left(\frac{1}{n} \right)
\end{array}
\label{information03}
\end{equation}
where $C_1(\alpha,\gamma)$ and $C_2(\alpha,\gamma)$ are constants that depend on $\alpha$ and $\gamma$ and can be found explicitly 
\begin{equation}
\begin{array}{l}
\displaystyle C_1= \frac{1}{\alpha+\gamma}+\frac{1}{2-\alpha-\gamma} - \log \frac{\alpha(2-\alpha-\gamma)}{(1-\alpha)(\alpha+\gamma)} \left( \frac{\gamma}{\alpha(\alpha+\gamma)} - \frac{1-\gamma}{(1-\alpha)(2-\alpha-\gamma)} \right),
\end{array}
\label{constant1}
\end{equation}
\begin{equation}
\begin{array}{l}
\displaystyle C_2= \left( \frac{1}{6 (-1 + \alpha)^2} -\frac{1}{6 (-2 + \alpha + \gamma)^2} - \frac{c}{3 \alpha (\alpha + \gamma)^2} - \frac{\gamma^2}{6 \alpha^2 (\alpha + \gamma)^2} \right) \log \frac{\alpha(2-\alpha-\gamma)}{(1-\alpha)(\alpha+\gamma)} + \\ 
\displaystyle \frac{-4 \alpha^6 - 8 \alpha^5 (-2 + \gamma) + (-2 + \gamma)^2 \gamma^2 - 4 \alpha (-2 +\gamma)^2 \gamma^2 + 
 \alpha^4 (-27 + 20 \gamma) + 4 \alpha^3 (6 - 4\gamma - 3\gamma^2 + 2 \gamma^3)}{4 (-1 + \alpha)^2 \alpha^2 (-2 + \alpha +
    \gamma)^2 (\alpha +\gamma)^2} + \\ \displaystyle \frac{ 
 \alpha^2 (-8 + 4\gamma + 22 \gamma^2 - 20\gamma^3 + 4 \gamma^4)}{4 (-1 + \alpha)^2 \alpha^2 (-2 + \alpha +
    \gamma)^2 (\alpha +\gamma)^2}.
\end{array}
\label{constant2}
\end{equation}

Also note that
\begin{equation}
 \begin{array}{c}  \displaystyle g_3(\alpha) = \int_{0}^{1} p \phi^{(n)}_3 f^{(n)}_\alpha {\rm d} p  =  \frac{\Gamma(2n+2) \Gamma (x+y+2)}{\Gamma(2n+3) \Gamma(x+y+1)}  = \\ \displaystyle \frac{\alpha+\gamma}{2} + \frac{1-\alpha-\gamma}{2n} + O \left(\frac{1}{n^2} \right). \end{array}
 \label{gfunction3}
 \end{equation}
 It is easy to see that in this case $g(\alpha)$ has different asymptotic comparing two cases above, so $g(\alpha) - \mathbb{E}(Z_\alpha)$ does not tend to zero as before.
Proceeding with the same computations as before we obtain 
\begin{equation}
\mathbb{V}^{\phi_3}(Z^{(n)}_\alpha) \geq \frac{(\alpha - \gamma)^2}{4} + C_3(\alpha,\gamma)\frac{1}{n} + O\left(\frac{1}{n^{3/2}} \right)
\label{var3}
\end{equation}
where $C_3(\alpha,\gamma)$ is a constant depending on $\alpha$ and $\gamma$ and can be found explicitly 
\begin{equation}
\begin{array} {l}
C_3= \displaystyle \frac{(\alpha-\gamma)}{48 (-1 + \alpha) \alpha (-2 + \alpha + \gamma) (\alpha + \gamma) ( \log(1 - \alpha) - \log(\alpha) - 
    \log(2 - \alpha - \gamma) + \log(\alpha + \gamma))} \times \\ (-48 \alpha^2 + 84 \alpha^3 - 24 \alpha^4 - 72 \alpha \gamma + 132 \alpha^2 \gamma - 72 \alpha^3 \gamma + 24 \gamma^2 - 
 12 \alpha \gamma^2 - 24 \alpha^2 \gamma^2 - 12 \gamma^3 + 24 \alpha \gamma^3 + \\ \displaystyle
    \frac{\left((-1 +\alpha) (39 \alpha^4 - 2 (-2 + \gamma) \gamma^2 - \alpha^3 (50 + 9\gamma) + 
   \alpha^2 (-56 + 146 \gamma - 135 \gamma^2) + \alpha \gamma (-44 + 194 \gamma - 87 \gamma^2)) \right)}{\left(\log(1-\alpha)-\log(\alpha)\right)^{-1}} + \\ \displaystyle \log \frac{\alpha+\gamma}{2-\alpha-\gamma} (56 \alpha^2 - 6 \alpha^3  - 89 \alpha^4  + 
 39 \alpha^5 + 44 \alpha \gamma  - 190 \alpha^2 \gamma  + 
 155 \alpha^3 \gamma  - 9 \alpha^4 \gamma  - 4 \gamma^2 - 
 190 \alpha \gamma^2 + \\ 329 \alpha^2 \gamma^2 - 
 135 \alpha^3 \gamma^2  + 2 \gamma^3 + 85 \alpha \gamma^3  - 
 87 \alpha^2 \gamma^3)  )
     
    \end{array}.
    \label{constant3}
    \end{equation} 
\end{proof}
\section{Weighted Bhattacharyya inequality}

\begin{theorem} \textbf{(Weighted Bhattacharyya inequality, uniparametric case).}\\
\noindent \textbf{(a)}
Let $\theta$ be a scalar parameter, $\tau(\theta)$ be a preassigned scalar function of parameter $\theta$. An unbiased estimator of $\tau(\theta)$ is a scalar function $T(\bZ)$ such that
\begin{equation}
e(\theta)=\mathbb{E}_\theta[T(\bZ)] = \tau(\theta).
\label{func}
\end{equation}
Consider the weight function that satisfies the condition (\ref{norm}). 
Recall
\begin{equation}
g(\theta) \equiv \int_{\mathbb{R}^d} T(\bz)  \phi (\bz, \theta, \gamma) f_\theta(\bz)  {\rm d} \bz.
\label{notation2}
\end{equation}
Assume that integrands in (\ref{notation2}) and (\ref{norm}) converge uniformly in $\theta$ after operation of differentiation up to order $\nu$.  Then the following inequality for the weighted variance of $T$ holds
\begin{equation}
\mathbb{V}^{\phi}_\theta(T) \geq \sum_{i,j=1}^{\nu} \left( g^{(i)}(\theta)-{\rm Q}_1^i + \tau {\rm Q}_2^i \right) \left( g^{(j)}(\theta)-{\rm Q}_1^j + \tau {\rm Q}_2^j \right) {\rm J}^{\phi}_{ij} 
\label{bhuni}
\end{equation}
where $Q_i^j$, $i=1,2$ are given in (\ref{Q1}) and (\ref{Q2}) respectively and $J_{ij}^\phi$ are the elements of the matrix $\mathbb{J}^\phi$ defined in (\ref{Jmat}).

\noindent \textbf{(b)} Consider RV $Z_\alpha^{(n)}$ with PDF $f^{(n)}_\alpha$ given in (\ref{pdf}) with $x=\lfloor \alpha n \rfloor$ where $0<\alpha<1$. When $\nu=2$, $\theta=\alpha$, $T(\bZ)=\bZ_\alpha^{(n)}$  for the weight function
\begin{equation}
\phi^{(n)}(p) = \frac{1}{\kappa_2(\alpha,\gamma)}p^{\gamma \sqrt{n}} (1-p)^{(1-\gamma)\sqrt{n}},
\label{weight2}
\end{equation}
inequality (\ref{bhuni}) takes the following form
\begin{equation}
\mathbb{V}^{\phi}_\theta(Z_\alpha^{(n)}) \geq \frac{C_4}{n} + \frac{C_5}{n^{3/2}}+ O \left( \frac{1}{n^2} \right)
\end{equation}
where $C_4$, $C_5$ are some constants that depend on $\alpha$ and $\gamma$ that given explicitly in (\ref{constant4}).
\end{theorem}

\begin{proof}
\noindent \textbf{(a)} 
Consider the function $R_\nu(\bz,\theta)$:
\begin{equation}
R_\nu(\bZ;\theta) = T(\bZ) - \tau(\theta) - \sum_{i=1}^{\nu} \lambda_i f_\theta^{(i)} f_\theta^{-1}
\label{Rfunction}
\end{equation}
where $\lambda_i$ are undefined parameters. It is easy to note that
\begin{equation}
\mathbb{E}[R_\nu(\bZ;\theta)]=0.
\label{expectation1}
\end{equation}
Consider the weighted variance given in (\ref{weight0}) of $R_\nu$. Because of (\ref{expectation1}) it can be written in the following form
\begin{equation}
\mathbb{V}^{\phi}_\theta(R_\nu) = \int_{\mathbb{R}^d} \left( T(\bz) - \tau(\theta) - \sum_{i=1}^{\nu} \lambda_i f_\theta^{(i)} f_\theta^{-1} \right)^2 \phi(\bz,\theta,\gamma) f_\theta {\rm d} \bz.
\label{varR}
\end{equation}
By the conditions of Theorem the differentiation is justified and leads to the following condition:
\begin{equation}
 \int_{\mathbb{R}^d} \left( T(\bz) - \tau(\theta) - \sum_{i=1}^{\nu} \lambda_i^{\star} f_\theta^{(i)} f_\theta^{-1} \right) \phi f_\theta^{(j)} {\rm d} \bz = 0.
 \label{foc}
\end{equation}
It can be rewritten as
\begin{equation}
\sum_{i=1}^{\nu} \lambda_i^{\star} \int_{\mathbb{R}^d} f_\theta^{(i)} f_\theta^{-1} f_\theta^{(j)} \phi {\rm d} \bz =   \int_{\mathbb{R}^d}  T(\bz)  \phi f_\theta^{(j)} {\rm d} \bz - \tau(\theta) \int_{\mathbb{R}^d}  \phi f_\theta^{(j)} {\rm d} \bz.
\label{foc2}
\end{equation}
Let $\mathbb{I}^\phi_\theta$ be the $\nu \times \nu$ matrix which elements are 
$$ I^{\phi}_{ij} = \int_{\mathbb{R}^d} f_\theta^{(i)} f_\theta^{-1} f_\theta^{(j)} \phi {\rm d} \bz $$ $i,j \leq \nu$. Let
\begin{equation}
\mathbb{J}^\phi_\theta = \left( \mathbb{I}^\phi_\theta \right)^{-1}
\label{Jmat}
\end{equation}
be the inverse $\nu \times \nu$ matrix and elements of this matrix are $J_{ij}^\phi$.

Note that in the case 
$i=j=1$, $I^{\phi}_{11}$ equals to the weighted Fisher information given in  (\ref{fisher}).

Consider integrals in RHS of (\ref{foc2}) separately. Firstly,
$$\int_{\mathbb{R}^d}  T(\bz) \tilde{\phi}  \left( \frac{1}{\kappa(\theta,\gamma)} f_\theta \right)^{(j)}  {\rm d} \bz = g^{(j)}(\theta),$$

$$\int_{\mathbb{R}^d}  T(\bz) \tilde{\phi}  \left[\sum_{k=0}^{j-1} {j\choose k} \left( \frac{1}{\kappa(\theta, \gamma)} \right)^{(j-k)} f_\theta^{(k)} \right] {\rm d} \bz + \int_{\mathbb{R}^d}  T(\bz) \phi f_\theta^{(j)}{\rm d} \bz = g^{(j)}(\theta).$$
Thus,
\begin{equation}
\int_{\mathbb{R}^d}  T(\bz) \phi f_\theta^{(j)}{\rm d} \bz = g^{(j)}(\theta) - {\rm Q}_1^j 
\end{equation}
where
\begin{equation} {\rm Q}_1^j = \int_{\mathbb{R}^d}  T(\bz) \tilde{\phi}  \left[\sum_{k=0}^{j-1} {j\choose k} \left( \frac{1}{\kappa(\theta,\gamma)} \right)^{(j-k)} f_\theta^{(k)} \right] {\rm d} \bz.
\label{Q1}
\end{equation}
In the analogous way from the condition (\ref{norm}) the following equality can be derived:
\begin{equation}
\int_{\mathbb{R}^d}  \phi f_\theta^{(j)}  {\rm d} \bz = - {\rm Q}_2^{j}
\label{Q2}
\end{equation}
where
\begin{equation} {\rm Q}_2^{j} = \int_{\mathbb{R}^d}  \tilde{\phi}  \left[\sum_{k=0}^{j-1} {j\choose k} \left( \frac{1}{\kappa(\theta,\gamma)} \right)^{(j-k)} f_\theta^{(k)} \right] {\rm d} \bz.
\label{Q2}
\end{equation}
So, (\ref{foc2}) takes the form
\begin{equation}
g^{(j)}(\theta) = \sum_{i=1}^{\nu} \lambda_i^{\star} I^{\phi}_{ij} + {\rm Q}_1^j - \tau {\rm Q}_2^j
\end{equation}
and
\begin{equation}
\lambda_i^{\star} = \sum_{j=1}^{\nu} \left( g^{(j)}(\theta)-{\rm Q}_1^j + \tau {\rm Q}_2^j \right)  J^{\phi}_{ij}.
\end{equation}
Thus, we obtain the following equality
\begin{equation}
\mathbb{V}(R^*_{\nu}) = \mathbb{V}^{\phi}_\theta(T) - \sum_{i,j=1}^{\nu} \left( g^{(i)}(\theta)-{\rm Q}_1^i + \tau {\rm Q}_2^i \right) \left( g^{(j)}(\theta)-{\rm Q}_1^j + \tau {\rm Q}_2^j \right) J^{\phi}_{ij}.
\end{equation}
The non-negativity of variance implies the lower bound for weighted variance of $T$ given in (\ref{bhuni}).
\begin{remark} Note that this inequality includes the weighted version of Rao-Cram\'er inequality. It appears when $\tau(\theta)=e(\theta)$, $\theta = \alpha$, $g(\theta) = g(\alpha)$, $T(\bZ)=\bZ$ and $i=j=\nu=1$. In this particular case
$$I^{\phi}_{11} = I^{\phi}(\theta) = \int_{\mathbb{R}^d} (f_\theta^{\prime})^2 f_\theta^{-1} \phi {\rm d} \bz,$$
$$\int_{\mathbb{R}^d}  \phi f_\theta^{(j)} {\rm d} \bz = \frac{\kappa^{\prime}(\theta,\gamma)}{\kappa(\theta,\gamma)}$$
and
$$\int_{\mathbb{R}^d}  T(\bz) \phi f_\theta^{(1)}{\rm d} \bz = g^{\prime}(\theta) + \frac{\kappa^{\prime}(\theta,\gamma)}{\kappa(\theta,\gamma)} g(\theta).$$
Thus, we obtain the inequality given in (\ref{rao}).
\end{remark}

\noindent \textbf{(b)} The lower bound in (\ref{bhuni}) takes the following form:
\begin{equation} \begin{array} {l}
\left( g^{(1)}(\theta)-{\rm Q}_1^1 + \tau {\rm Q}_2^1 \right) \left(J^{\phi}_{12} + J^{\phi}_{21} \right)\left( g^{(2)}(\theta)-{\rm Q}_1^2 + \tau {\rm Q}_2^2 \right) + \\  +  \left( g^{(2)}(\theta)-{\rm Q}_1^2 + \tau {\rm Q}_2^2 \right)^2 J^{\phi}_{22}  +  \left( g^{(1)}(\theta)-{\rm Q}_1^1 + \tau {\rm Q}_2^1 \right)^2 J^{\phi}_{11} 
 \end{array}
 \label{bh}
\end{equation}
where $J^{\phi}_{ij}$ are $ij^{th}$ elements of the matrix $\mathbb{J}^\phi_\theta$ defined in (\ref{Jmat}). Moreover, the asymptotic of $I^{\phi}_{11}$ is given above. Compute the asymptotic of other terms.
$$ I^{\phi}_{12} = \int_{\mathbb{R}} f^{(1)} f^{-1} f^{(2)} \phi {\rm d} \bz = L_1 n^{3/2} + L_2n + L_3 \sqrt{n} + L_4 + O \left( \frac{1}{\sqrt{n}} \right)$$
where $L_i$ $i=1,2,3,4$ are the constants that can be found explicitly and dependent only $\alpha$ and $\gamma$, but have very large construction,
$$ L_1=\frac{\alpha^3+\alpha^2(-2+\gamma) + \alpha (2-3 \gamma) \gamma + \gamma^3}{(1-\alpha)^3 \alpha^3}.$$
$$L_2 = \frac{-2 \alpha^5 - 3 \gamma^4 + \alpha^4 (-3 + 16 c) + 2 \alpha^3 (5 - 17 \gamma + \gamma^2)}{2(1-\alpha)^4 \alpha^4}+$$
 $$+\frac{2 \alpha \gamma^2 (-4 + 3 \gamma + 3 \gamma^2) + \alpha^2 (-2 + 6 \gamma + 24 \gamma^2 - 18 \gamma^3)}{2(1-\alpha)^4 \alpha^4}$$
$$L_3 = \frac{-21 \gamma^5 + 24 \alpha^6 (-1 + 2 \gamma) + \alpha^5 (13 + 24\gamma - 168 \gamma^2)-2 \alpha^2 \gamma^3 (-109 + 72 \gamma + 36 \gamma^2) }{12 (-1 + \alpha)^5 \alpha^5}$$
$$\frac{ \alpha \gamma^3 (-44 + 33 \gamma + 72 \gamma^2) + 
 \alpha^4 (44 - 237 \gamma + 492 \gamma^2 - 48 \gamma^3) + 
 6 \alpha^3 (-2 + 10\gamma - 19 \gamma^2 - 56 \gamma^3 + 36 \gamma^4)}{12 (-1 + \alpha)^5 \alpha^5}$$
 $$L_4 = \frac{16 \alpha^9 - 15 \gamma^6 + \alpha^8(40 \gamma - 92) - 4\alpha^6(-41+14\gamma + 26 \gamma^2) + 2 \alpha \gamma^4 (-12+10\gamma + 35 \gamma^2) }{8(1-\alpha)^6 \alpha^6} +$$
 $$+ \frac{\alpha^6(-161 + 118\gamma + 136\gamma^2 + 152\gamma^3) + \alpha^2 \gamma^2 (24-12\gamma + 157 \gamma^2 - 120 \gamma^4)}{8(1-\alpha)^6 \alpha^6}+$$
$$\frac{2 \alpha^5 (66-158\gamma + 97 \gamma^2 - 308 \gamma^3 + 40 \gamma^4 + \alpha^4(-52+148\gamma + 75 \gamma^2 + 160\gamma^3 + 400 \gamma^4 - 240 \gamma^5)}{8(1-\alpha)^6 \alpha^6}+$$
$$ + \frac{ 4 \alpha^3(2-6\gamma - 25\gamma^2+4\gamma^3-97\gamma^4+60\gamma^4+20\gamma^6)}{8(1-\alpha)^6 \alpha^6}$$

The asymptotic of  $I^{\phi}_{22}$ takes the form
$$I^{\phi}_{22} = \int_{\mathbb{R}} (f^{\prime \prime})^2 f^{-1} \phi {\rm d} p= L_5 n^{2}+L_6 n^{3/2} + L_7n + L_8 \sqrt{n} + L_9 + O \left( \frac{1}{\sqrt{n}} \right)$$ where $L_i$  are some constants again that can be found explicitly and depend on $\alpha$ and $\gamma$. In order to compute $I^{\phi}_{22}$, one need to compute the integral of the following form
$$\int_0^1 {\rm log}(1-p)^i {\rm log}(p)^j p^{A_1(\alpha,\gamma)n + A_2(\alpha,\gamma) \sqrt{n}} (1-p)^{A_3(\alpha,\gamma)n + A_4(\alpha,\gamma) \sqrt{n} } {\rm d} p$$ for $i=1,2,3,4$ and $j=1,2,3,4$ which were computed above for cases $i=1,2$ and $j=1,2$ and one can compute the integral for larger $i$ and $j$ by integration by parts. The only problem with deriving the exact coefficients is the computational cost, so we proceed in terms of constants $L_i$.

In order to use the same notation we will write $I^{\phi}_{11}$ in the following form:
$$I^{\phi}_{11} = \int_{\mathbb{R}} (f^{\prime})^2 f^{-1} \phi {\rm d} p= L_{10}n+ L_{11} \sqrt{n} + L_{12} + O \left( \frac{1}{\sqrt{n}} \right)$$
where coefficients $L_{10},L_{11}, L_{12}$ are found above.

Other terms in (\ref{bh}) can be computed explicitly. Using the notations of previous section we write
$$Q_1^1 = - \frac{\kappa^{\prime}(\alpha,\gamma)}{\kappa(\alpha,\gamma)} g(\alpha),$$
$$Q_1^2 = - \left( \frac{1}{\kappa(\alpha,\gamma)} \right)^{\prime \prime} \int_0^1 p \tilde{\phi}f dp - 2\frac{\kappa^{\prime}(\alpha,\gamma)}{\kappa(\alpha,\gamma)} \left( g^{(1)} - Q_1^1 \right),$$ 
$$Q_2^1 = - \frac{\kappa^{\prime}(\alpha,\gamma)}{\kappa(\alpha,\gamma)},$$
and
$$Q_2^2 = \left( \frac{1}{\kappa(\alpha,\gamma)} \right)^{\prime \prime} \int_0^1 \tilde{\phi}f dp -2\frac{\kappa^{\prime}(\alpha,\gamma)}{\kappa(\alpha,\gamma)}Q_2^1  $$
where 
$ g^{(1)}$ is given in (\ref{gder2}).
Thus, we obtain the following asymptotic of lower bound for the weighted variance in stated Bayesian problem:
\begin{equation}
\mathbb{V}^{\phi}(T) \geq \frac{C_4}{n} + \frac{C_5}{n^{3/2}}+ O \left( \frac{1}{n^2} \right)
\end{equation}
where $C_4$ and $C_5$ are some constants that depend on $\alpha$ and $\gamma$ and can be  found explicitly, but they also have too cumbersome construction. As an example $C_4$ is given below:
\begin{equation}
\begin{array}{l}
\displaystyle
C_4=\frac{2((\alpha-\gamma)^2 + \alpha(1-\alpha)) (-2\alpha^2 + \alpha^3 + 2 \alpha \gamma + \alpha^2 \gamma - 3 \alpha \gamma^2 + \gamma^3)L_1}{(1-\alpha)^3 \alpha^3 (L_1^2 - L_{10}L_5)}\\ + \displaystyle \frac{(-2\alpha^2 + \alpha^3 + 2\alpha \gamma + \alpha^2 \gamma - 3 \alpha \gamma^2 + \gamma^3)^2 L_{10}}{(1-\alpha)^4 \alpha^4 (-L_1^2 + L_{10}L_{5})} + \frac{\left(1+ \frac{(\alpha-\gamma)^2}{(1-\alpha)\alpha} \right)^2 L_5}{-L_1^2 + L_{10}L_5}.
\end{array}
\label{constant4}
\end{equation}

\begin{remark}
Note that in the case $\alpha=\gamma$ the first and second term in $C_4$ vanish. Also one can easily check that $L_1=0$ in this case. So, because of $L_{10} = \frac{1}{\alpha(1-\alpha)}$ we have
$$C_4= \frac{1}{L_{10}} = \alpha(1-\alpha).$$
Thus, the main term of asymptotic is exactly the same as was obtained above in the standard Cram\'er-Rao case.
\end{remark}

\end{proof}

\begin{theorem}\textbf{ (Weighted Bhattacharyya inequality, multiparametric case).}
Let $\theta \in \Theta \subset R^m$ be a vector of parameters, $\tau(\theta) = \left( \tau_1(\theta), \ldots, \tau_l(\theta) \right)^{\rT} \in \mathbb{R}^l$ be  the preassigned vector function of parameter $\theta$ and $T(\textbf{Z})$ be an unbiased estimate of $\tau(\theta)$:
$$e(\theta)=\mathbb{E}_\theta(T) = \int_{R^d} T(\textbf{z}) f_\theta(\bz) {\rm d} \bz = \tau(\theta).$$
Consider the weight function $\phi(\bz, \theta, \gamma)$ such that the condition (\ref{norm}) holds. Assume that the following positively definite matrix exists
\begin{equation}
I^{\phi} = \mathbb{E}^{\phi}_\theta[\beta \beta^{\rT}]
\end{equation}
where
$$\beta = (\beta_1(\theta),\ldots, \beta_r(\theta))^{\rT}$$
is $r$-dimensional RV, components of which are all possible expressions of the following form
\begin{equation}
\frac{1}{f_\theta(\textbf{Z})}\frac{\partial^{i_1, \ldots, i_m}}{\partial \theta_1^{i_1}, \ldots \partial \theta_m^{i_m}} f_\theta(\textbf{Z})
\label{beta}
\end{equation}
 where $(1\leq i_1 + \ldots + i_m \leq s)$ and $r$ is the total number of all these expressions.

Let $\mathbb{F}^{\phi}$ be the ($r \times l$) matrix which rows has the following form
\begin{equation}
\int_{\mathbb{R}^d} \left(T(\textbf{z}) - \tau (\theta) \right) \phi(\bz,\theta,\gamma) \frac{\partial^{i_1, \ldots, i_m}}{\partial \theta_1^{i_1}, \ldots \partial \theta_m^{i_m}}  f_\theta(\textbf{z})  {\rm d} {\bz}
\label{notation1}
\end{equation}
numbered in the same order as expressions (\ref{beta}). Assume that integrands in (\ref{notation1}) and (\ref{norm}) converge uniformly in $\theta$ after the operation of differentiation. 
Then the following inequality for weighted variance of $T$ holds
\begin{equation}
\mathbb{V}^{\phi}_\theta(T)   \geq (\mathbb{F}^{\phi})^{\rT} I^{\phi}
(\theta)^{-1} \mathbb{F}^{\phi}.
\label{bhmulti}
\end{equation}
\end{theorem}
 
 \begin{remark}
Here and below for ($d \times d$) matrices of the same dimension $d$,  $\mathbb{A}$ and $\mathbb{B}$, the inequality
 $$\mathbb{A} \geq \mathbb{B}$$
 means that $$\mathbb{C}=\mathbb{A}-\mathbb{B}$$ is a non-negatively definite matrix.
 \end{remark}

\begin{proof}
Note that elements of matrix $\mathbb{F}^{\phi}$ can be found from the condition (\ref{norm}).

Consider one dimensional RV
$$\delta = [ (T - \tau) - \beta^{\star}(I^{\phi})^{-1} \mathbb{F}^{\phi}]y$$
where $y^{\rT} = (y_1,\ldots, y_l) \in \mathbb{R}^l$ is a non-random vector. It is easy to see that $\mathbb{E}_\theta(\delta) = 0$. 
Taking weighted expectation of both sides in equality 
\begin{equation}
\delta^2 = y^{\rT} \left[ (T-\tau)(T-\tau)^{\star} - 2(T-\tau) \beta^{\star} (I^{\phi})^{-1} \mathbb{F}^{\phi} + (\mathbb{F}^{\phi})^{\star} (I^{\phi})^{-1}\beta \beta^* (I^{\phi})^{-1}\mathbb{F}^{\phi} \right] y,
\label{expre}
\end{equation}
for any $y$ we obtain
\begin{equation}
\mathbb{E}^\phi_\theta(\delta^2) = y^{\rT} \left[ \mathbb{V}^{\phi}_\theta(T)   - (\mathbb{F}^{\phi})^{\rT} (I^{\phi})^{-1} \mathbb{F}^{\phi} \right] y.
\end{equation}
The non-negativity of variance implies the multi-parametric version of Bhattacharyya inequality, given in (\ref{bhmulti}). One can easily see that in uni-parametric and 1D case this inequality equivalent to the weighted Cram\'er-Rao inequality.
\end{proof}

\section{Weighted Kullback inequality}
\begin{theorem}\textbf{(Weighted Kullback inequality)}\\
\textbf{(a)} 
For given PDFs $f$,$g$
\begin{equation}
K^\phi(f||g) \geq \Psi^*_{\tilde{g}}(\mu_\phi(\tilde{f})) = \sup_{t} \left[ \langle t, \mu_\phi(f) \rangle + {\rm log}C(g) - {\rm log} \bar{M}_g(t) \right]
\label{kullback}
\end{equation}
where
\begin{equation}\bar{M}_g(t) = \int_{\mathbb{R}^d}\phi(\bz) e^{ \langle t, \bz \rangle } g(\bz) {\rm d}{\bz}
\end{equation}
is a weighted moment generating function, $t \in \mathbb{R}^d$ and
$$\mu_\phi({f})= \frac{\mathbb{E}_f[\bZ \phi (\bZ)]}{\mathbb{E}_f[\phi(\bZ)]} \in \mathbb{R}^d $$
is the classical expectation of $\tilde{f}$.

\textbf{(b)} Let $Z_\alpha^{(n)}$ and $Z_\rho^{(n)}$ be RVs with PDF $f_\alpha^{(n)}$ given in (\ref{pdf}) with $x= \lfloor \alpha n \rfloor $ and with PDF $f_\rho^{(n)}$ given in (\ref{pdf})  with $x= \lfloor \rho n \rfloor $ respectively where $0<\alpha,\rho<1$   and weight function
\begin{equation}
\phi^{(n)}(p) = \frac{1}{\kappa(\rho,\gamma)} p^{\gamma \sqrt{n}} (1-p)^{(1-\gamma)\sqrt{n}},
\label{weight002}
\end{equation}
where $ \kappa(\rho,\gamma)$ is found from normalization condition
\begin{equation}
 \int_{0}^1 \phi^{(n)} f_\rho^{(n)}  {\rm d} p = 1.
 \label{norm01}
 \end{equation}
Denote $\epsilon = \alpha - \rho$ then
$$K^\phi(f_\alpha^{(n)}||f_\rho^{(n)}) \geq \frac{\epsilon^2 \left(1+\sqrt{n}-n \right)^2}{2(1-\alpha)\alpha n} +  O \left( 1 \right). $$
As $\epsilon \to 0$,
\begin{equation}
\exists \lim_{\epsilon \to 0}\frac{1}{\epsilon^2} K^\phi(f_\alpha^{(n)}||f_\rho^{(n)}) = \frac{1}{2}I(\tilde{f_\alpha}) \geq \frac{n}{2\alpha (1-\alpha)}-  \frac{\sqrt{n}}{\alpha (1-\alpha)} + O \left( 1 \right)
\label{kullback2}
\end{equation}
where $I(\tilde{f}_\alpha)$ is the standard Fisher information.
\end{theorem}

\begin{proof}
\noindent \textbf{(a)}  The inequality (\ref{kullback}) is proved in \cite{Kelbert2015}. 

\noindent \textbf{(b)}  Firstly, note that by (\ref{norm01}):
$${\rm log} \left(C(f_\rho^{(n)}) \right) = 0.$$
The weighted generating function of RV $Z_\rho^{(n)}$ with PDF $f_\rho^{(n)}$ equals:
$$\bar{M}_{f_\rho^{(n)}}(t) = \int_0^1 \phi^{(n)} e^{tp} f_\rho^{(n)} {\rm d}p ={_1F_1}(\rho n + \gamma \sqrt{n} + 1, n + \sqrt{n} + 2; t)$$
$$ = 1 + \sum_{k=1}^\infty \frac{t^k}{k!}\prod_{j=0}^{k-1} \frac{\rho n + \gamma \sqrt{n} + 1+j}{ n + \sqrt{n} + 2+j}$$
where $_1F_1(x, y; z)$ is the confluent hypergeometric function.

For large $n$, the expression for weighted generating function can be written in the following way \cite[formula 12]{hapaev}: 

$$\bar{M}_{f_\rho^{(n)}}(t) = 1 + \sum_{k=1}^\infty \frac{t^k}{k!}\prod_{j=0}^{k-1} \frac{\rho n + \gamma \sqrt{n} + 1+j}{ n + \sqrt{n} + 2+j} = $$

$$= \sum_{k=0}^\infty \frac{t^k}{k!} \left( \rho^k - k (\rho^k-\rho^{k-1} \gamma) \frac{1}{\sqrt{n}} + \frac{\rho^{k-2}k(\rho-2\rho^2-\gamma^2+\rho k - 2\rho \gamma k + \gamma^2 k)}{2n}+ O \left(\frac{1}{n^{3/2}} \right) \right) $$
$$ = e^{\rho t} \left( 1 - (\rho-\gamma) t \frac{1}{\sqrt{n}} + \frac{2(1-\rho-\gamma)t+(\rho-2\rho \gamma + \gamma^2)t^2}{2n} + O \left(\frac{1}{n^{3/2}} \right) \right).$$
Thus, we have that
$${\rm log}\bar{M}_{f_\rho^{(n)}}(t) = \rho t + {\rm log} \left( 1 - (\rho-\gamma) t \frac{1}{\sqrt{n}} + \frac{2(1-\rho-\gamma)t+(\rho-2\rho \gamma + \gamma^2)t^2}{2n} + O \left(\frac{1}{n^{3/2}} \right) \right)=$$
$$ =  \rho t - (\rho - \gamma)t \frac{1}{\sqrt{n}} + \frac{(1-\rho-\gamma)t+ \frac{\rho t^2}{2}(1-\rho)}{n} + O \left(\frac{1}{n^{3/2}} \right). $$

The first term in (\ref{kullback}) for PDF $f_\alpha^{(n)}$ and weight function $\phi^{(n)}$ takes the following form 
$$ \mu_\phi({f_\alpha^{(n)}}) = \frac{\alpha n + \gamma \sqrt{n} +2}{n + \sqrt{n} +2} = \alpha+ (\gamma-\alpha) \frac{1}{\sqrt{n}} + \frac{1-\alpha-\gamma}{n} + O \left( \frac{1}{n^{3/2}} \right).$$ 
Then
\begin{equation}
 \Psi^*_{f_\rho}(\mu_\phi(\tilde{f_\alpha}))= \sup_t \left[ (\alpha - \rho) t -  (\alpha - \rho) \frac{t}{\sqrt{n}} + \frac{\rho-\alpha}{n}t - \frac{(1-\rho)\rho}{2n}t^2 + O \left( \frac{1}{n^{3/2}} \right) \right].
 \label{psistar}
 \end{equation}
 Finding supremum of the expression above, we obtain
$$\tau= \frac{(\alpha-\rho) \left(n-1-\sqrt{n} \right)} {(1-\alpha)\alpha} + O \left( \frac{1}{n^{1/2}} \right).$$
So
\begin{equation}
 \Psi^*_{f_\rho}(\mu_\phi(\tilde{f_\alpha})) =  \frac{(\alpha-\rho)^2 \left(1+\sqrt{n}-n \right)^2}{2(1-\alpha)\alpha n}+ O \left( 1 \right).
 \label{psistar}
 \end{equation}
Denote $\epsilon = \alpha -\rho$. When $\epsilon \to 0$ we obtain
$$\frac{1}{\epsilon^2} \Psi^*_{f_\rho}(\mu_\phi(\tilde{f_\alpha})) = \frac{n}{2\alpha (1-\alpha)}-  \frac{\sqrt{n}}{\alpha (1-\alpha)}+ O \left( 1 \right).$$ Thus
\begin{equation}
\exists \lim_{\epsilon \to 0} \frac{1}{\epsilon^2} K^\phi(f_\alpha^{(n)}||f_\rho^{(n)}) = \frac{1}{2} I(\tilde{f}_\alpha) \geq \frac{n}{2\alpha (1-\alpha)}-  \frac{\sqrt{n}}{\alpha (1-\alpha)} + O \left( 1 \right)
\label{kullback02}
\end{equation}
which completes the proof of Theorem 5.

\end{proof}

\section*{Acknowledgement}
The article was prepared within the framework of a subsidy granted to the HSE by the Government of the Russian Federation for the implementation of the Global Competitiveness Program.

\end{document}